\documentclass[12pt]{article}
\usepackage{amsmath,bm,amssymb,xcolor,multirow,mathtools,hyperref,nameref,relsize,amsthm,placeins,boldline,babel}
\newcommand\myeq{\stackrel{\mathclap{\normalfont\mbox{a.e.}}}{=}}
\usepackage{graphicx}
\newtheorem{theorem}{Theorem}
\newtheorem{assump}{Assumption}
\usepackage{float}
\usepackage{caption}
\usepackage{tikz}
\usetikzlibrary{shapes,arrows}
\usepackage{multicol}
\usepackage{setspace}
\newcommand\undermat[2]{%
	\makebox[0pt][l]{$\smash{\underbrace{\phantom{%
					\begin{matrix}#2\end{matrix}}}_{\text{$#1$}}}$}#2}
\usepackage{subfig}
\usetikzlibrary{fit}
\usepackage[top=2cm, bottom=2cm, left=2cm, right=2cm]{geometry}

\DeclareMathOperator*{\argmin}{arg\,min}
\usepackage[utf8]{inputenc}
\newcommand*{\affmark}[1][*]{\textsuperscript{#1}}
\usepackage{authblk}
\title{A Computationally Efficient algorithm to estimate the Parameters of a Two-Dimensional Chirp Model with the product term}
\author{Abhinek Shukla\protect\affmark[1]\affmark[,3], Rhythm Grover\affmark[2], Debasis Kundu\affmark[1], and Amit Mitra\affmark[]}
\affil[1]{Department of Mathematics and Statistics, Indian Institute of Technology Kanpur, Kanpur - 208016, India}
\affil[2]{Mehta Family School of Data Science and Artificial Intelligence, Indian Institute of Technology Guwahati, Assam-781039, India}
\affil[3]{Corresponding author. Email: abhushukla@gmail.com}
\date{}
\begin{document}
\maketitle
	\begin{abstract}
Chirp signal models and  their generalizations have been used to model many natural and man-made phenomena in signal processing and time series literature. In recent times, several methods have been proposed for parameter estimation of these models. These methods however are either statistically sub-optimal or computationally burdensome, specially for two dimensional (2D) chirp models.  In this paper, we consider the problem of parameter estimation of 2D chirp models and propose a computationally  efficient estimator and establish asymptotic theoretical properties of the proposed estimators.
And the proposed estimators  are observed to have the
same rates of convergence as the least squares estimators (LSEs).  Furthermore, the proposed estimators of chirp rate parameters
are shown to be asymptotically optimal. Extensive and detailed numerical simulations are conducted,  which support theoretical results of the proposed estimators. 
	\end{abstract}
		\text{ \\~\\~\\ }{\bf Keywords:}
%alphabetical order
		2D Chirp model, least squares estimators, stationary linear process, consistency, asymptotic normality.
		%	\MSC[2020] Primary 62H12 \sep
		%	Secondary 62F12

	\section{Introduction}
			This paper addresses the problem of parameter estimation of a 2D chirp signal model defined as follows:
	\begin{align}\label{model_2D}
		y(m,n) &=A^0\cos(\alpha^0 m+\beta^0 m^2+\gamma^0 n+\delta^0 n^2+\mu^0mn)\nonumber\\&+B^0\sin(\alpha^0 m+\beta^0 m^2+\gamma^0 n+\delta^0 n^2+\mu^0mn)+X(m,n),\\& \hspace{15pt} m=1,2,\ldots,M, n=1,2,\ldots,N.\nonumber
	\end{align}	
	Here, $ y(m,n)$ is the observed real valued signal and $X(m,n)$ is the additive noise term.  $A^0,B^0$ are amplitude parameters, $\alpha^0,\gamma^0$ are frequency parameters, $\beta^0,\delta^0$ are frequency rates or chirp rates, and $\mu^0$ is the coefficient of product term. $\mathrm{\bm{\xi}}^0 = (\alpha^0,\beta^0 , \gamma^0 ,  \delta^0 ,   \mu^0)^\top \hspace{2pt}$ represents vector of non-linear parameters.  This model can be used to describe signals having  constant amplitude with frequency to be a linear function of spatial co-ordinates. The product term \(mn\) in such chirp models \eqref{model_2D}, is an important characteristic of numerous measurement interferometric signals or radar signal returns. \\~\\
 The parameter estimation problem for model \eqref{model_2D} is encountered in many real-life applications such as    2D-homomorphic signal processing, magnetic resonance imaging (MRI), optical imaging, interferometric	synthetic aperture radar (INSAR), modeling non-homogeneous patterns  in the texture image captured by a camera due to perspective or orientation ( see e.g., \cite{Francos_1995},\cite{Francos_1998},  \cite{Francos_1999} and the references cited therein). 2D chirps have been used as a spreading function in digital watermarking which is also helpful in   data security, medical safety,  fingerprinting, and observing content manipulations, see \cite{Zhang_watermark}.
 2D chirp signals have also been used to model Newton's rings \cite{Guo_Newton_Rings}.  These rings are predominantly used in
 testing spherical and flatting optical surface and  curvature radius measurement. \\~\\
	Many algorithms based on different approaches have been put forward in the literature to solve such problems.  Polynomial phase differencing (PD) operator was introduced in  \cite{Francos_1995} as  an extension of the polynomial phase transform proposed in \cite{Peleg_1991}. Several works \cite{Francos_1996}, \cite{Francos_1998} and \cite{Francos_1999} utilized PD operator to develop computationally efficient algorithms
	for estimating similar polynomial phase signals.  Cubic phase function (CPF) proposed in \cite{O'Shea_2002_CPF},  was extended in  \cite{Zhang_2008_CPF} for similar 2D chirp signal modelling. 	Further,   CPF was  utilized to estimate  2D cubic phase signal using genetic algorithm in \cite{Djurovic_2010_genetic}.   Consistency and asymptotic normality of LSEs for a general  2D polynomial phase signal (PPS) model have been derived in  \cite{Lahiri_2017}. A  finite step computationally efficient procedure
	for a similar 2D chirp signal model proposed in  \cite{Lahiri_2012} was proved asymptotically equivalent to
	LSEs.   Quasi-Maximum Likelihood (QML) algorithm \cite{Djurovic_2014} proposed for 1D PPS, was generalized for 2D PPS in \cite{Djurovic_2017}.  Further approximate least squares estimators (ALSEs)  proposed in \cite{Grover_2d_2018}  have been proved to be asymptotically equivalent to LSEs. An efficient estimation procedure based on fixed dimension technique, presented in \cite{Grover_efficient}  was shown to be equivalent to the optimal LSEs, for a 2D chirp model without the product term \(mn\). \\~\\
	Estimators based on phase differencing strategies or high order ambiguity function (HAF) or  some of their modifications
	are computationally easier to obtain. However, the performance of estimation  deteriorate 
	below a relatively high signal-to-noise ratio (SNR) threshold and are sub-optimal. Methods that use  PD in the steps of estimation, usually estimate coefficients of the highest
	degree first, and then subsequently  estimate the coefficients of a lower degree from the demodulated or de-chirped signal. Therefore, the
	estimation error of highest degree coefficients accumulates and affects estimation accuracy of lower degree
	coefficients quite seriously. For more details, one can refer to \cite{Barbaross_1998}, \cite{Djurovic_2017} and \cite{Wu_2008}.  Till date, there is no detailed study of the theoretical properties of the estimators such as CPF and QML, such as strong consistency and asymptotic normality. Recently, optimal estimators for a simpler 2D chirp model without the interaction term  have  been developed in  \cite{Grover_efficient}. However, the results in \cite{Grover_efficient} cannot be generalized directly for the underlying model \eqref{model_2D}. It may be noted that the model considered in this paper is more general, as it takes into account the interaction term $\mu^0 mn$. Due to the presence of this interaction term coefficient $\mu^0$, the estimation becomes more difficult as the estimators of $\alpha^0$ and $\gamma^0$ are no longer independent (as in the case for \cite{Grover_efficient}), and hence making their computation as well as study of theoretical analysis becomes more challenging. 
	The problem becomes more complicated under the assumption of general stationary linear process error assumption.    \\~\\
	The main contributions of this paper are;  providing a computationally efficient algorithm to estimate the parameters of the model defined in \eqref{model_2D}, and further establishing theoretical asymptotic properties of the proposed  estimators.
	The proposed algorithm is motivated by the fact that a 2D chirp model with five non-linear parameters can be viewed through a number of  1D chirp models with two non-linear parameters and hence  computational complexity of 2D models can be reduced.\\~\\
	 The key attributes of the proposed method are that it is computationally faster than the conventional optimal methods such as LSEs, maximum likelihood estimators, or ALSEs and at the same time, having desirable statistical properties such as, attaining same rates of convergence as the optimal LSEs. In fact, the proposed estimators of the chirp rate parameters, then these  have the same asymptotic variance as that of the traditional LSEs, and hence are statistically optimal.\\~\\
 The rest of the paper is organised as follows: the methodology to obtain the proposed  estimators is presented in Section \ref{methd}. The model assumptions and the asymptotic  theoretical results are given  in Section \ref{resul}.   
	In Section \ref{simul}, the finite sample performance of the proposed estimators is demonstrated through simulation studies. In this section, a  comparison of the performance of the proposed  estimators with the state-of-the-art methods such as least squares method, approximate least squares method, and 2D multilag HAF method is also presented.
	Finally, Section \ref{conclu} concludes the paper, followed by  detailed proofs in appendices.
	\section{Estimation Methodology}\label{methd}
In this section, we present the proposed method of estimation. 	Let the data matrix for model \eqref{model_2D} be denoted as 
$$\bm{Y} = \begin{bmatrix}
	y(1,1)& y(1,2)&\ldots &y(1,N)\\
	y(2,1)& y(2,2)& \ldots &y(2,N)\\
	\vdots&\vdots &\ddots &\vdots\\
	y(M,1)&y(M,2)&\ldots& y(M,N)
\end{bmatrix}_{M\times N}.$$
The proposed algorithm uses the fact that for each fixed column (or row) of \(\bm{Y}\), the 2D chirp model breaks down to a cascade of 1D chirp models.\\~\\
Realise that if we fix one dimension say $n=n_0$ in \eqref{model_2D}, then the 2D chirp can be seen as 1-D chirp for $m=1,2,\ldots,M$, as follows:
	\begin{align} \label{decomp_1}
		y(m,n_0) = A^0(n_0)\cos\big((\alpha^0 +n_0\mu^0)m+\beta^0 m^2 \big)&+B^0(n_0)\sin\big((\alpha^0 +n_0\mu^0)m+\beta^0 m^2 \big)\nonumber\\&+X(m,n_0),\end{align}
\vskip -2 cm	\begin{align*}
		\text{where, }\hspace{10pt}A^0(n_0) &= A^0\cos( \gamma^0 n_0+\delta^0 n_0^2)+B^0\sin( \gamma^0 n_0+\delta^0 n_0^2),\\ \hspace{45pt}
		B^0(n_0) &= -A^0\sin( \gamma^0 n_0+\delta^0 n_0^2)+B^0\cos( \gamma^0 n_0+\delta^0 n_0^2).
	\end{align*}
	Similarly, for a fixed $m=m_0$, we have 1-D chirp for $n=1,2,\ldots,N,$
	\begin{align}\label{decomp_2}
		y(m_0,n) = \widetilde{A}^0(m_0)\cos\big((\gamma^0+m_0\mu^0) n+\delta^0 n^2 \big)&+\widetilde{B}^0(m_0)\sin\big((\gamma^0+m_0\mu^0) n+\delta^0 n^2 \big)\nonumber \\&+X(m_0,n).%\\\nonumber~\\&\text{where, }\hspace{10pt}A^0_{\alpha}(m_0) = A^0\cos( \alpha^0 m_0+\beta^0 m_0^2)+B^0\sin( \alpha^0 m_0+\beta^0 m_0^2)\nonumber\\&\hspace{45pt}
		%	B^0_{\alpha}(m_0) = -A^0\sin( \alpha^0 m_0+\beta^0 m_0^2)+B^0\cos( \alpha^0 m_0+\beta^0 m_0^2)\nonumber
	\end{align}
Equation \eqref{decomp_1} represents 1-D chirp signal model with $\alpha^0+n_0\mu^0$ and  $\beta^0$ as the frequency and frequency rate parameters respectively. Similarly, equation \eqref{decomp_2} represents 1-D chirp signal model with $\gamma^0+m_0\mu^0$ and  $\delta^0$ as the frequency and frequency rate parameters respectively.\\~\\
 Hence, our methodology is developed by estimating parameters of these 1D chirps based on a particular column (or row) vector of data matrix, rather than estimating the whole 2D chirp parameters based on the full data matrix. Therefore this procedure reduces computational burden to estimate model parameters drastically.
Further suppose column vector  $\bm{Y}_{Mn_0}$ denotes the $n_0^{th}$ column of data matrix $\bm{Y}$ and column vector  $\bm{Y}_{m_0N}$ denotes the transpose of $m_0^{th} $  row of data matrix $\bm{Y}$. Define $Z_k(\alpha_1,\alpha_2)$ matrix as 
\begin{equation}Z_k(\alpha_1,\alpha_2)=
	\begin{bmatrix}
		\cos(\alpha_1+\alpha_2) & \sin(\alpha_1+\alpha_2)\\
		\cos(\alpha_12+\alpha_22^2) & \sin(\alpha_12+\alpha_22^2)\\
		\vdots&\vdots\\
		\cos(\alpha_1k+\alpha_2k^2) & \sin(\alpha_1k+\alpha_2k^2)\\
	\end{bmatrix}_{k\times 2}.
\end{equation} 
We need to estimate the model parameters of \eqref{decomp_1} which is a 1D chirp model, so we use LSEs to estimate the parameters. We obtain LSEs of non-linear parameters in model \eqref{decomp_1} for a fixed $n_0$ by defining following reduced sum of squares, see \cite{Grover_efficient}:
 \begin{equation}
R_{Mn_0}(\alpha_1,\alpha_2) = \bm{Y}_{Mn_0}^{T}\Big(I_{M\times M}-P_{Z_M}(\alpha_1,\alpha_2)\Big)\bm{Y}_{Mn_0},
\end{equation}
where, $P_{Z_M}(\alpha_1,\alpha_2) = Z_M(\alpha_1,\alpha_2)\Big(Z_M(\alpha_1,\alpha_2)^\top Z_M(\alpha_1,\alpha_2)\Big)^{-1}Z_M(\alpha_1,\alpha_2)^\top $ and $I_{M\times M}$ is the $M\times M$ identity matrix.\
Then, \begin{equation}\label{express_alp_bet}
	 (\widehat{\alpha}_{n_0},\widehat{\beta}_{n_0})^\top  = {\displaystyle\argmin_{\alpha_1,\alpha_2}}R_{Mn_0}(\alpha_1,\alpha_2)
\end{equation} 
is  the proposed estimator of $(\alpha^0+n_0\mu^0,\beta^0)^\top $  based on minimizing the sum of squares corresponding to $n_0^{th}$ column of the data matrix $\bm{Y}$. Similarly, we can obtain  LSEs of parameters in model \eqref{decomp_2} by defining 
\begin{equation}
R_{m_0N}(\alpha_1,\alpha_2) = \bm{Y}_{m_0N}^{T}\Big(I_{N\times N}-P_{Z_N}(\alpha_1,\alpha_2)\Big)\bm{Y}_{m_0N},
\end{equation}
where, $P_{Z_N}(\alpha_1,\alpha_2) = Z_N(\alpha_1,\alpha_2)\Big(Z_N(\alpha_1,\alpha_2)^\top Z_N(\alpha_1,\alpha_2)\Big)^{-1}Z_N(\alpha_1,\alpha_2)^\top $ and $I_{N\times N}$ is the $N\times N$ identity matrix. 
Then, \begin{equation} \label{express_gam_delt}
	(\widehat{\gamma}_{m_0},\widehat{\delta}_{m_0})^\top  = {\displaystyle\argmin_{\alpha_1,\alpha_2}}R_{m_0N}(\alpha_1,\alpha_2)
	\end{equation} 
 will be the proposed estimator of $(\gamma^0+m_0\mu^0,\delta^0)^\top $  based on minimizing the sum of squares corresponding to  $m_0^{th}$ row of data matrix $\bm{Y}$. \\~\\
  We observe that for each fixed column, we get an estimate of the same chirp rate parameter $\beta^0$ in \eqref{express_alp_bet}, and also an estimate of frequency parameter which is a linear combination of $\alpha^0$ and $\mu^0$. Similarly, estimates of $\delta^0$ and a linear combination of $\gamma^0$ and $\mu^0$ for a fixed row in \eqref{express_gam_delt} has been obtained. It is important to note that the linearity of parameters of 1D-chirp models plays a crucial role in getting proposed estimators of $\alpha^0,\gamma^0$ and $\mu^0$ by fitting a linear regression model as follows.\\~\\
Once the parameters corresponding to each $(M+N)$ 1-D chirp models have been estimated. We apply the following three steps to obtain final estimates of parameters of the model \eqref{model_2D}:
\begin{enumerate}
	\item[Step-1.] Let $\bm{\Gamma}^\top =\begin{bmatrix}
		1&1&\cdots &1&0&0&\cdots&0\\
		0&0&\cdots&0&1&1&\cdots&1\\
		1&2&\cdots&N&1&2&\cdots&M
	\end{bmatrix}$, and  $\bm{\Lambda}^\top =\begin{bmatrix}
		\widehat{\alpha}_1&
		\widehat{\alpha}_2&
		\cdots&
		\widehat{\alpha}_N&
		\widehat{\gamma}_1&
		\widehat{\gamma}_2&
		\cdots&
		\widehat{\gamma}_M
	\end{bmatrix}$. \\
	Combine the obtained estimates as follows:
	\begin{equation} \label{lin_eq}
		\bm{\Lambda}=\bm{\Gamma}\begin{bmatrix}
			\alpha\\
			\gamma\\
			\mu
		\end{bmatrix}.
	\end{equation}
	 Then  estimate of $(\alpha^0,\gamma^0,\mu^0)^\top$   is $\Big(\bm{\Gamma}^\top \bm{\Gamma}\Big)^{-1}\bm{\Gamma}^\top \bm{\Lambda}$.
	\item [Step-2.] The estimates of  $\beta^0$ and $\delta^0$ are  simply the averages $\displaystyle\widehat{\beta}= \cfrac{1}{N}\sum_{n=1}^{N}\widehat{\beta}_{n}$ and $\displaystyle\widehat{\delta} =\cfrac{1}{M}\sum_{m=1}^{M}\widehat{\delta}_{m}$, \\respectively.
	\item[Step-3.] After getting estimates of non-linear parameters, $\widehat{\mathrm{\bm{\xi}}} = (\widehat{\alpha},\widehat{\beta},\widehat{\gamma},\widehat{\delta}, \widehat{\mu})^\top \hspace{2pt}$, the amplitude parameter estimates can be provided as follows: 
\end{enumerate} \begin{equation}
	\begin{bmatrix}
		\widehat{A}\\\widehat{B}
	\end{bmatrix} =\begin{bmatrix}
	\cfrac{2}{{MN}}\displaystyle \sum_{m=1}^{M}\sum_{n=1}^{N}y(m,n)\cos\widehat{\phi} \\
	\cfrac{2}{{MN}}\displaystyle \sum_{m=1}^{M}\sum_{n=1}^{N}y(m,n)\sin\widehat{\phi}
\end{bmatrix}.
\end{equation}
\section{Theoretical Results}\label{resul}
In this section, we first state the model assumptions required to derive the theoretical asymptotic properties explicitly. These are as follows:
\begin{assump}\label{asp_1}
	$X(m,n)$ can be expressed as a linear combination of a double array sequence of independently and identically distributed (i.i.d.) random variables $\{\epsilon(m,n)\}$ with mean $0$, variance $\sigma^2$ and  finite fourth moment.
	\begin{equation} \label{err_1}
		X(m,n) = \displaystyle{\sum_{i=-\infty}^{\infty}\sum_{j=-\infty}^{\infty}}a(i,j)\epsilon(m-i,n-j), \end{equation}  
	such that
	\begin{equation} \label{err_2}
		\hspace{5pt}\displaystyle{\sum_{i=-\infty}^{\infty}\sum_{j=-\infty}^{\infty}}|a(i,j)|<\infty. \end{equation} 
\end{assump}
\begin{assump}\label{asp_2}
  True parameter $	\mathrm{\bm{\theta}}^0 = (A^0,B^0,\alpha^0,\beta^0,\gamma^0,\delta^0,\mu^0)^\top  $ is an interior point of parameter space $\Theta$, where	$\Theta = (-\infty,\infty)\times(-\infty,\infty)\times[0,2\pi]\times[0,\pi/2]\times[0,2\pi]\times[0,\pi/2]\times[0,2\pi]$, and  $A^{0^2}+B^{0^2}>0$.
\end{assump}
Assumption \ref{asp_1} puts the model under a very general set-up of  noise contamination as it includes the dependent relationship too. Assumption \ref{asp_2} is taken to assure the absence of any identifiability problem and non-zero deterministic part of the signal. Under these general assumptions, we have derived strong consistency and  asymptotic normality of the estimators. The obtained results are stated in the following theorems. 
\begin{theorem}\label{thm_1}
	Under assumptions \ref{asp_1} and \ref{asp_2}, the proposed  estimator of parameter $\mathrm{\bm{\theta}}^0$ is strongly consistent, i.e.,  $$\widehat{\mathrm{\bm{\theta}}}\xrightarrow{a.s.}\mathrm{\bm{\theta}}^0 \mbox{\hspace{10pt} as }   
	\min\{M,N\}\xrightarrow{}\infty.$$ 
\end{theorem}
\begin{proof}
	Please see \nameref{AppA} for the proof.
\end{proof}
\begin{theorem}\label{thm_2}
Under assumptions 1 and 2, the proposed   	estimators of $\mathrm{\bm{\theta}}^0$ is  asymptotically normally distributed. 	$$\bm{D}^{-1}(\widetilde{ \mathrm{\bm{\theta}}} -  \mathrm{\bm{\theta}}^0) \xrightarrow{d} \mathcal{ N}_7(0,\displaystyle \bm{\Sigma})\mbox{\hspace{10pt} as } M=N\rightarrow\infty , $$ where $c = \displaystyle{\sum_{i=-\infty}^{\infty}\sum_{j=-\infty}^{\infty}}a(i,j)^2$,  \hspace{2pt} $\bm{D}^{-1} = diag(M^{\frac{1}{2}}N^{\frac{1}{2}}, \hspace{5pt}M^\frac{1}{2}N^{\frac{1}{2}}, \hspace{5pt}
	M^{\frac{3}{2}}N^{\frac{1}{2}}, \hspace{5pt}M^{\frac{5}{2}}N^{\frac{1}{2}}, \hspace{5pt}M^{\frac{1}{2}}N^{\frac{3}{2}}, \hspace{5pt}M^{\frac{1}{2}}N^{\frac{5}{2}}, \hspace{5pt}M^{\frac{3}{2}}N^{\frac{3}{2}} \hspace{5pt} )$, and \\
	\begin{equation*}
		\bm{\Sigma}=
		\cfrac{c\sigma^2}{(A^{0^2}+B^{0^2})}\begin{bmatrix}
			2A^{0^2}+187B^{0^2}&-185A^0B^0&-378B^0&60B^0&-378B^0&60B^0&612B^0\\
			-185A^0B^0&2B^{0^2}+187A^{0^2}&	378A^0&-60A^0&378A^0&-60A^0&-612A^0\\
			-378B^0	&378A^0	&996&-360&612&0&-1224\\
			60B^0&-60A^0	&-360&360&0&0&0\\
			-378B^0	&378A^0	&612&0&996&-360&-1224\\
			60B^0&-60A^0	&0&0&-360&360&0\\
			612B^0	&-612A^0	&-1224&0&-1224&0&2448
		\end{bmatrix}.	
	\end{equation*}
	Here  $diag(a_1, a_2, ..., a_k)$ represents  $k\times k$ diagonal matrix with elements $a_1, a_2, ..., a_k$ in the principal diagonal and $\mathcal{ N}_k(\bm{\mathcal{M}},\bm{\mathcal{S}})$ represents the $k$-variate normal distribution with mean vector $\bm{\mathcal{M}}$ and variance-covariance matrix $\bm{\mathcal{S}}$. 
\end{theorem}
\begin{proof}
Please see \nameref{AppB} for the proof.
\end{proof}
Although Theorem \ref{thm_2} has been proved for the increasing sample size assuming $M=N\rightarrow \infty$, however asymptotic normality  will still hold even if $M/N\rightarrow p $ as $M,N\rightarrow \infty$, for some $p>0$. It is interesting to note that the asymptotic properties of the proposed estimators of chirp rates $\beta^0,\delta^0 $ will remain the same even if we take \(	\min\{M,N\}\xrightarrow{}\infty\).  The asymptotic variance-covariance matrix of  $(\alpha^0,\gamma^0,\mu^0)$ will however change  depending on the value of $p$ among non-linear parameters. \\~\\
If we further assume that the errors in \eqref{model_2D} are i.i.d. Gaussian distributed random variables, then it can be observed that the proposed estimators of chirp rates parameters $\beta^0$ and $\delta^0 $, asymptotically attain Cramer-Rao lower bound (CRLB). CRLB for estimators of other  non-linear parameters $\alpha^0$, $\gamma^0$, and $\mu^0$ are \(
	\cfrac{456c\sigma^2}{(A^{0^2}+B^{0^2})}\), \hspace{5pt} \(	\cfrac{456c\sigma^2}{(A^{0^2}+B^{0^2})}\)  \hspace{5pt}  and  \hspace{5pt}  \( 	\cfrac{288c\sigma^2}{(A^{0^2}+B^{0^2})}
\), see \cite{Lahiri_2017}.
\section{Simulation Results}\label{simul}
Simulation studies done in this paper are divided into three parts. The first part demonstrates the evaluation of finite sample size performance of proposed  estimators. We  compare the performance of the proposed estimators with  the  asymptotically optimal estimators such as LSEs, and ALSEs, and fast but sub-optimal 2D-multilag-HAF estimators. The second part shows the lower computational cost of the proposed estimators as compared to the LSEs. Finally, the third part exemplifies the ability of the proposed estimators to  extract original gray-scale texture from one contaminated with noise and reproduce the original texture.   We have performed simulations on the complex counterpart of the model \eqref{model_2D}, (see \cite{Barbaross_2014}) for comparison purposes. 
\subsection{Finite Sample Performance}
To provide a detailed assessment of the performance of proposed estimators, we have chosen sample sizes $M=N=20,40,60,80$ and 100. The fixed values of all parameters to obtain complex-valued chirp data are: 	\begin{equation} \label{par_choice}
		A^0= 1, \hspace{7pt}\alpha^0 =0.4 , \hspace{12pt}\beta^0=0.1429, \hspace{7pt}\gamma^0 =0.25 ,\hspace{7pt} \delta^0 = 0.1250, \hspace{7pt}\mu^0 =  0.1667.
\end{equation}                       
 Obtained data from the model  is then contaminated with noise $X(m,n)$. We consider two distinct noise structures for our simulations. These are:
\begin{itemize}
	\item Independently and identically distributed (i.i.d.) normal errors with mean 0 and variance $\sigma^2$;
	\item Autoregressive moving average (ARMA) errors with following representation:
	\begin{align}\label{err_cas_2}
		X(m,n) =& 0.06X(m-1,n-1)-0.054X(m,n-1)+0.087X(m-1,n)+\nonumber \\&\epsilon(m,n)+.01\epsilon(m-1,n-1)+0.035\epsilon(m,n-1)+0.042\epsilon(m-1,n).
	\end{align}
	where $\epsilon(m,n)$ is a sequence of i.i.d. Gaussian random variables with mean 0 and variance $\sigma^2$.  
\end{itemize}  
       We have obtained the estimates for 1000 replications  for a fixed sample size $M=N$, under a particular error structure with fixed $\sigma^2$.
         The estimators do not have any explicit closed form expression, so we  use  Nelder–Mead algorithm ( using ``optim"  function in R software) for optimization of the objective function and to obtain the estimators. Mean square errors (MSEs) obtained under 1000 replications are displayed in the figures \ref{fig:a1}, \ref{fig:b1}, \ref{fig:a2} and \ref{fig:b2} for four different values of $\sigma$  = 0.1,0.5,0.9,1. The MSEs are plotted on negative logarithmic scale.   The findings of  these simulation results  can be summarized as follows:
        \begin{itemize}
        \item MSEs   of the proposed estimators decrease rapidly as sample size increases, which supports the consistency property of the estimators.   Further, as sample size increases, the gap between the  MSEs of the proposed estimators and LSEs decreases.
        \item The obtained MSEs of proposed estimators of $\beta^0$ and $\delta^0$,   are at par with those of LSEs and ALSEs.
      
        \end{itemize}
    \subsection{Time Comparison}
    The computational advantage of the proposed estimators over the conventional LSEs is quite significant. In order to compare the two methods, we measure their computational complexities in terms of the number of points in the grid that are needed to find the initial guesses of these estimators. Once we have the precise initial guesses, applying an iterative algorithm like Nelder-Mead is a matter of seconds.  ``gridSearch" function from the package ``NMOF" is used  to calculate the initial guesses.  We report observed time to get the estimates for a fixed sample size and the total number of grid points over which cost function evaluations are required, in table \ref{tim_comp_1}. The choice of parameters is same as in \eqref{par_choice} along-with i.i.d. normal errors with mean 0 and standard deviation $\sigma=0.9$. For a fixed sample size $M=N$, the order of computation for LSEs is $M^4N^4=M^8$, but for the proposed method, the order of computation is $M^3N+N^3M=2M^4$.
    The numerical experiments for comparing time efficiency were performed on a system with processor: Intel(R) Core(TM) i3-5005U CPU @ 2.00GHz 2.00GHz; installed memory(RAM): 4.00 GB; and system type: 64-bit Operating System. Codes were written and run in R version 4.0.4 (2021-02-15) -- ``Lost Library Book", a free software environment for statistical computing and graphics. 
    \FloatBarrier
			\begin{figure}[]%[htbp!]
	\centering
	\subfloat[Standard deviation of $\epsilon(m,n)$ is $\sigma=0.1$ ]{\label{fig:a1}{\includegraphics[width=0.9\linewidth]{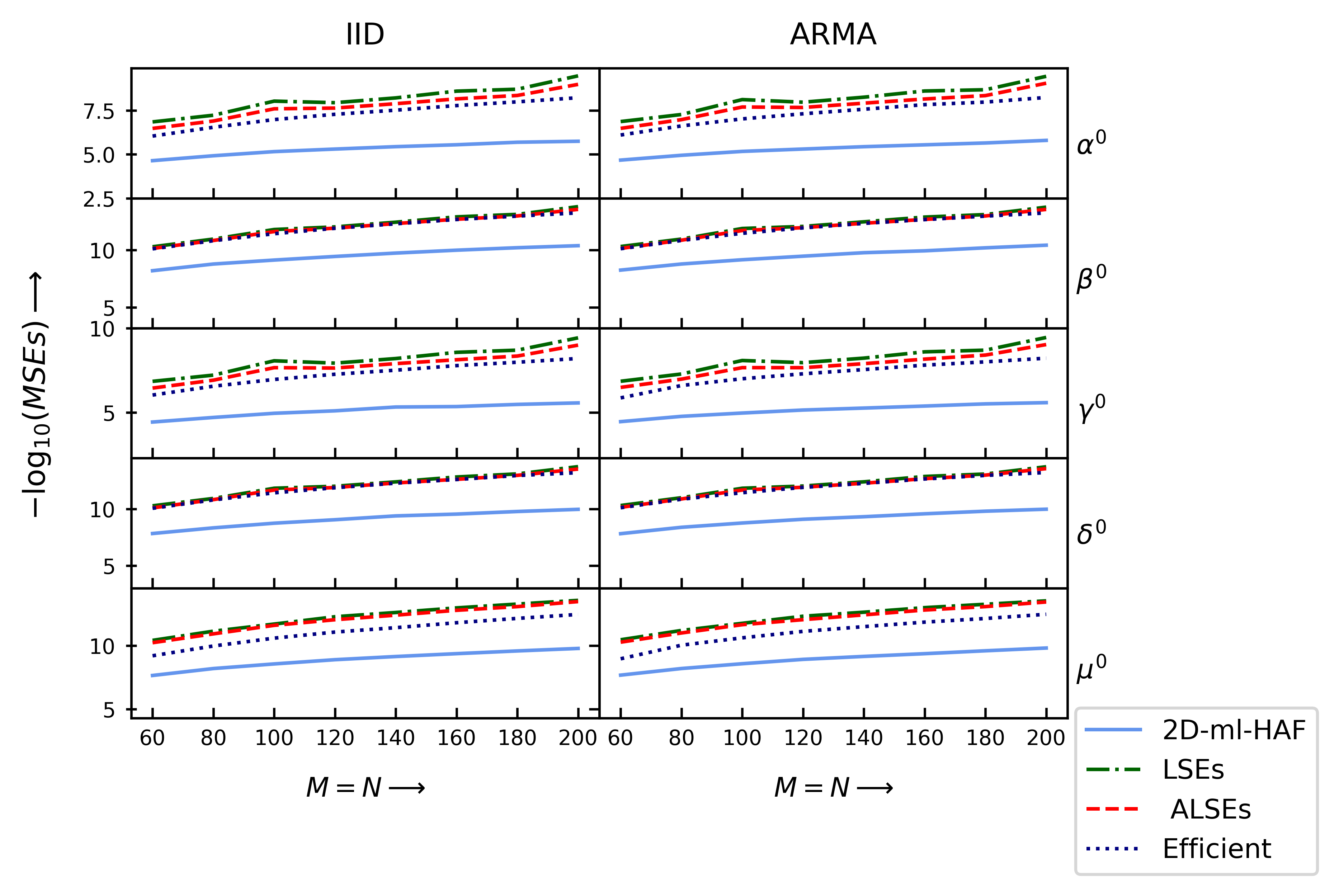}}}\\
	\subfloat[ Standard deviation of $\epsilon(m,n)$ is $\sigma=0.5$ ]{\label{fig:b1}{\includegraphics[width=0.9\textwidth]{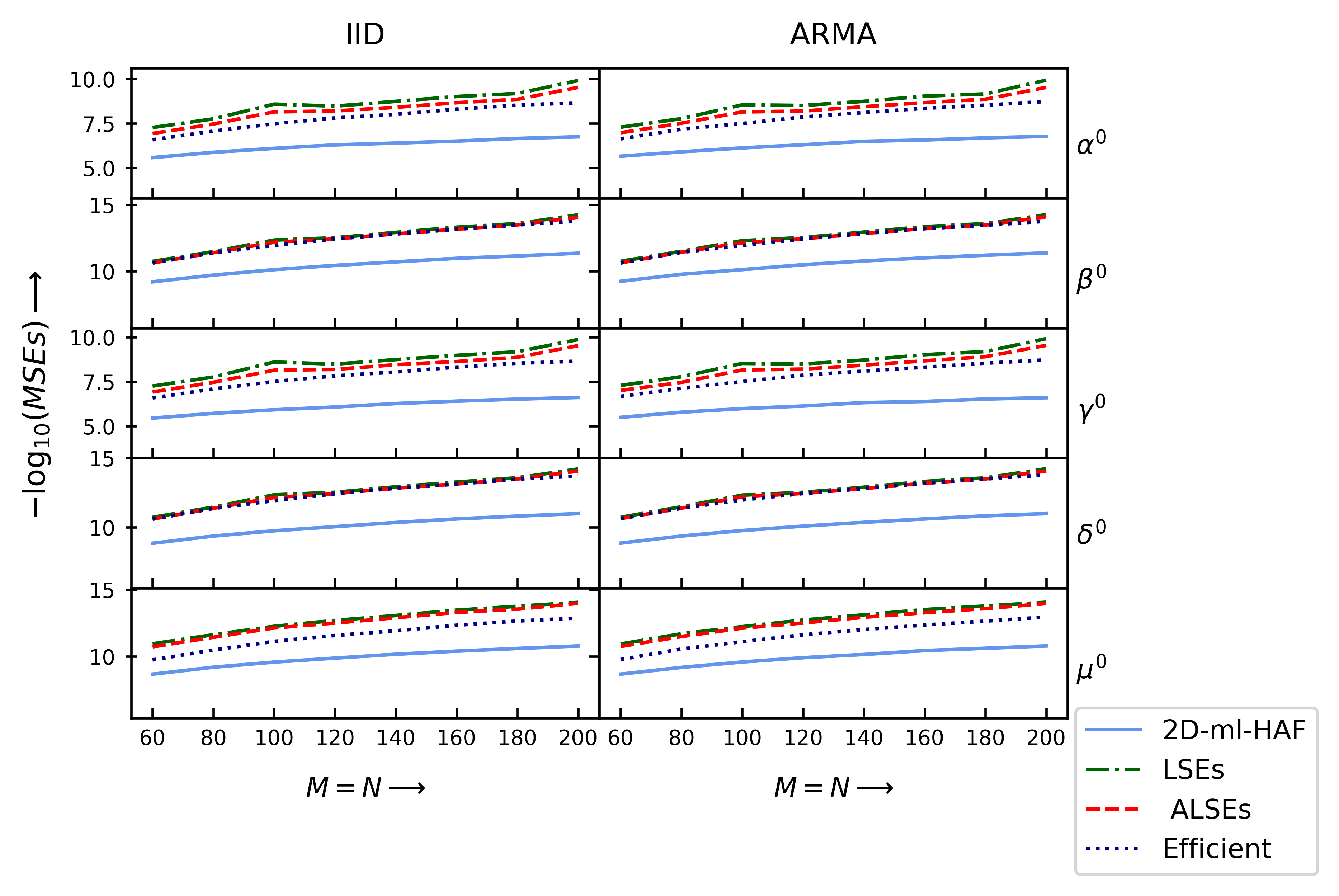}}}
	\label{fig:myfig1}
	\caption{Plots of $-\log(\mbox{MSEs})$ versus the sample size for estimators of non-linear parameters.}
\end{figure}
\newpage
\begin{figure}[]%[htbp!]
	\centering
	\subfloat[Standard deviation of $\epsilon(m,n)$ is $\sigma=0.9$ ]{\label{fig:a2}{\includegraphics[width=0.9\linewidth]{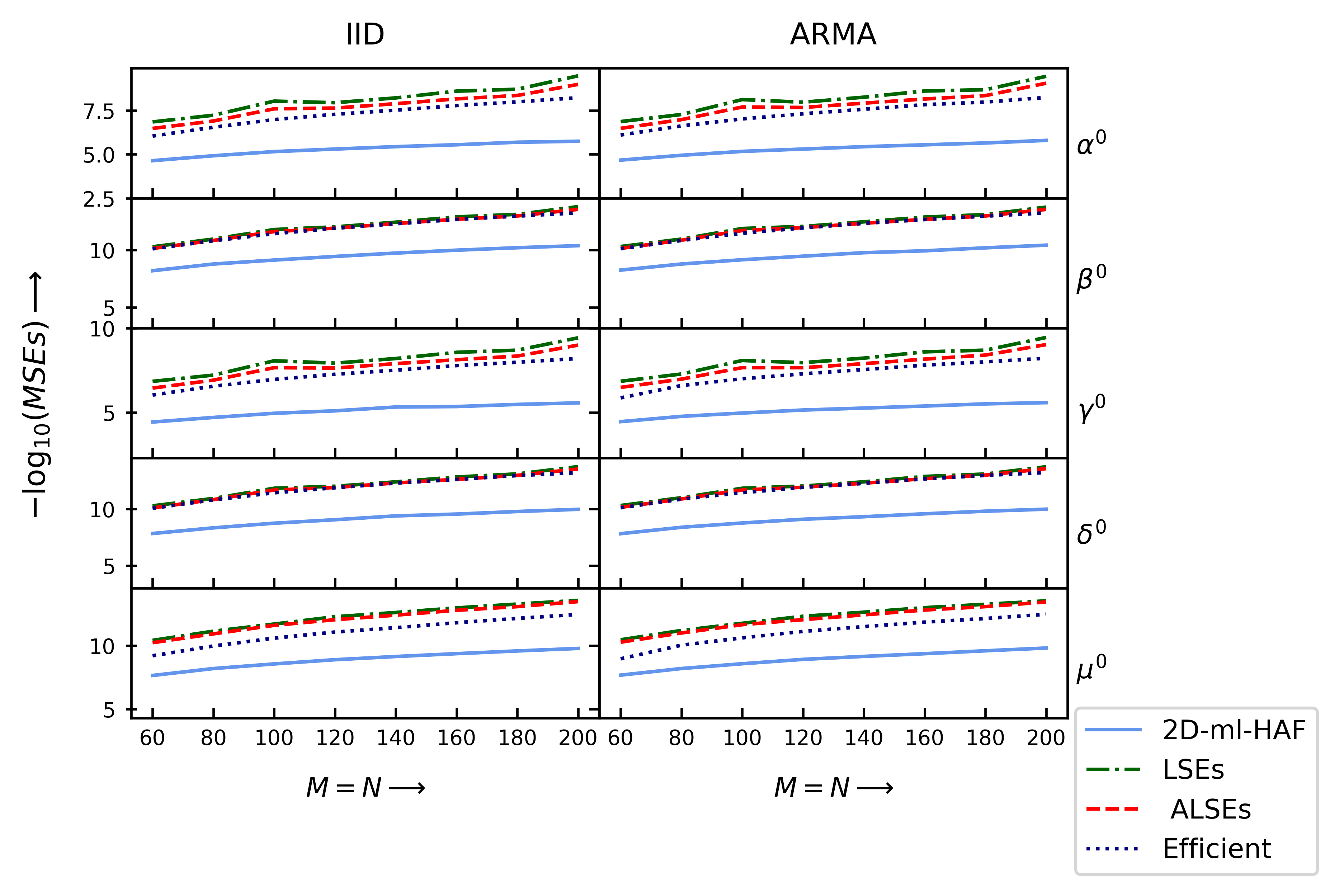}}}\\
	\subfloat[ Standard deviation of $\epsilon(m,n)$ is $\sigma=1$ ]{\label{fig:b2}{\includegraphics[width=0.9\textwidth]{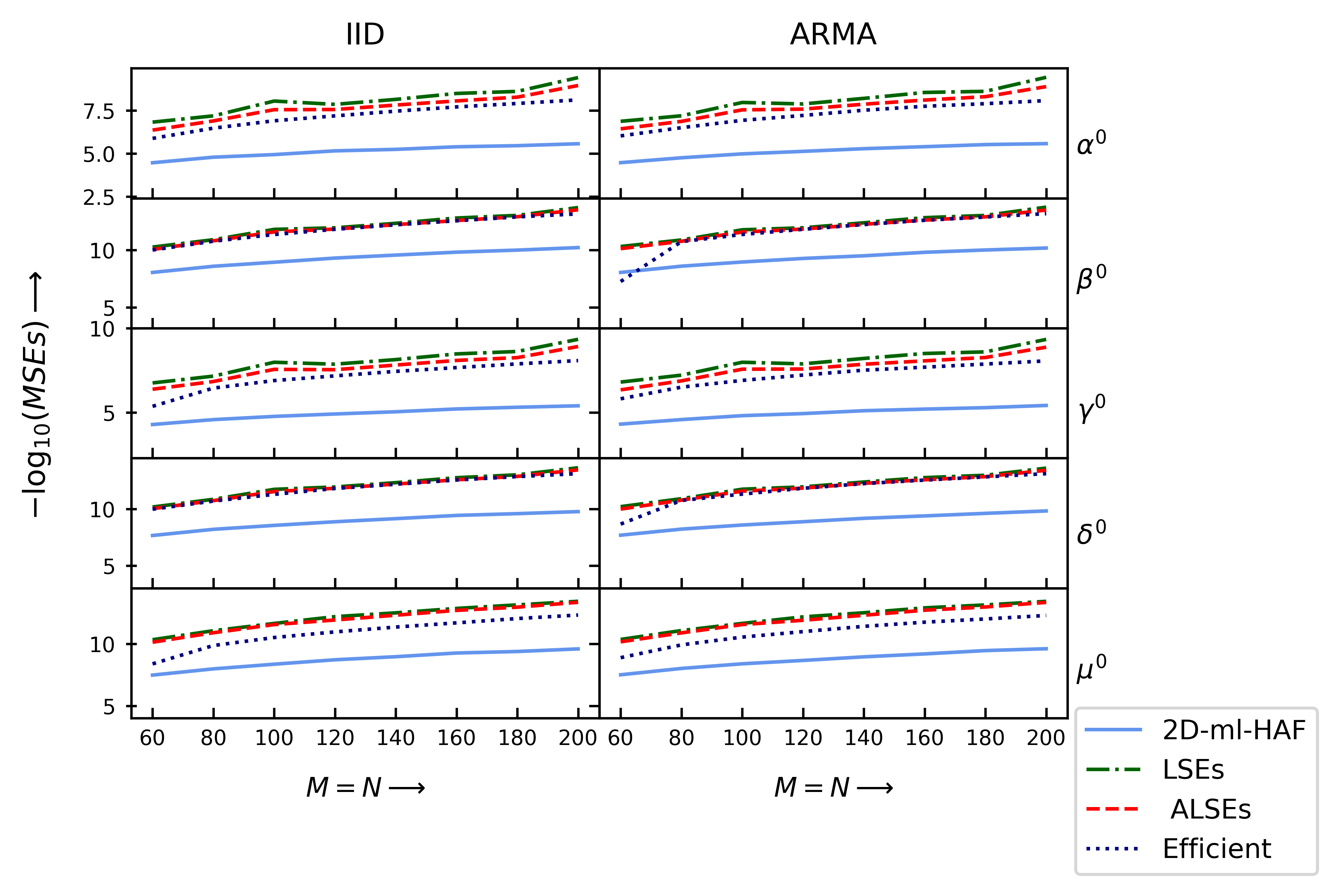}}}
	\label{fig:myfig2}
		\caption{Plots of $-\log(\mbox{MSEs})$ versus the sample size for estimators of non-linear parameters.}
\end{figure}
\FloatBarrier
 
\FloatBarrier
\begin{table}[h!]

	\footnotesize
	\caption{Time and number of grid points taken to compute the estimates }\label{tim_comp_1}
	\begin{center}
\scalebox{0.9}{	\begin{tabular}{|c||c|c||c|c|}

\hline					
\multicolumn{1}{|c||}{Sample Size}	&\multicolumn{2}{c||}{	Efficient }	&\multicolumn{2}{c|}{LSEs}	\\ \hline
$M=N$&Time(seconds)& Total No. of grid points&Time(seconds)& Number of grid points\\ \hline
2	&	0.044	&	12	&	0.128	&	27	\\ \hline
3	&	0.041	&	96	&	0.308	&	2048	\\ \hline
4	&	0.069	&	360	&	4.101	&	30375	\\ \hline
5	&	0.119	&	960	&	34.609	&	221184	\\ \hline
6	&	0.173	&	2100	&	198.038	&	1071875	\\ \hline
7	&	0.250	&	4032	&	856.375	&	3981312\\	\hline \hline
	\end{tabular}
}
\end{center}

\end{table}
\FloatBarrier
For the considered machine, it was not feasible to perform grid-search to get LSEs for more than  $M=N=7$. When we plot sample size $M(=N)$ against the time to compute initial guess for  LSEs, then it is observed to be linear. So, to get an idea of the time deviation of LSEs with that of proposed estimates at larger sample sizes,  we predict time to obtain LSEs based on grid-search by fitting a simple linear regression model between $\log$ of sample size and $\log$ of time to get LSEs. From the results in Table \ref{tab_pred}, we can clearly observe the massive time difference of getting proposed estimates and LSEs. For example, if we go for sample size, say $M=N=40$, then it will take more than 20 years to obtain LSEs using grid-search over the same machine (even if we assume a large amount of RAM in a machine), while it took less than 10 minutes to obtain the proposed estimates.
\FloatBarrier
\begin{table}[h!]
	\footnotesize
	
	\caption{Comparing Time efficiency with Predicted time for LSEs }\label{tab_pred}
	\begin{center}
	\scalebox{0.9}{	\begin{tabular}{|c||c|c||c|c|}

			\hline					
			\multicolumn{1}{|c||}{Sample Size}	&\multicolumn{2}{c||}{	Efficient }	&\multicolumn{2}{c|}{LSEs}	\\ \hline
			$M=N$& Computed Time & Total No. of grid points& Predicted Time & Number of grid points\\ \hline
8	&	0.530	sec	&	7.06E+03	&	32.165	min	&	1.23E+07		\\ \hline
9	&	0.734	sec	&	1.15E+04	&	1.375	hr	&	3.28E+07		\\ \hline
10	&	1.545	sec	&	1.78E+04	&	3.195	hr	&	7.86E+07		\\ \hline
15	&	5.831	sec	&	9.41E+04	&	3.411	days	&	2.20E+09		\\ \hline
20	&	26.479	sec	&	3.03E+05	&	34.070	days	&	2.29E+10		\\ \hline
25	&	1.015	min	&	7.49E+05	&	203.051	days	&	1.40E+11		\\ \hline
30	&	2.443	min	&	1.56E+06	&	2.392	yr	&	6.11E+11		\\ \hline
35	&	5.369	min	&	2.91E+06	&	8.209	yr	&	2.12E+12		\\ \hline
40	&	8.525	min	&	4.99E+06	&	23.888	yr	&	6.22E+12		\\ \hline
45	&	14.530	min	&	8.02E+06	&	61.288	yr	&	1.61E+13		\\ \hline
50	&	24.243	min	&	1.22E+07	&	142.37	yr	&	3.75E+13		\\ \hline \hline

		\end{tabular}
	}
	\end{center}
\end{table}
\FloatBarrier
\subsection{ Texture Pattern Estimation}
2D-chirp signals create interesting gray-scale texture patterns. In order to analyze the effectiveness of the proposed  estimators for estimating texture patterns accurately, we  generate data from the complex counterpart of the model with same set of parameters as in  \eqref{par_choice}. Then  real and imaginary part of the obtained data is contaminated independently with  i.i.d. normal errors having mean 0 and variance $\sigma^2=0.09$. The data matrix obtained is of size $100\times 100$. We analyze this data using the proposed  estimator, the optimal LSEs, ALSEs, and also with the 2D-multilag-HAF method. Note that we have used true values as initial guess for obtaining optimal estimators, LSEs and ALSEs because of the computational complexity and grid-search method for obtaining 2D-multilag-HAF estimators and proposed estimators. \\~\\
Plugging these estimators in the deterministic part of the model, we get  estimated texture patterns as the real part of the reproduced data. We also present  real part of the original dataset to compare the original and estimated textures. It is clear from the figures that texture pattern obtained using the proposed method is visually the same as that obtained using optimal LSEs and ALSEs, while the 2D-multilag-HAF estimator gives a slightly different pattern than the original one.

	\begin{figure}[htbp!]
	\centering
	\subfloat[Contaminated Texture
	]{\label{fig:a}{\includegraphics[width=0.22\linewidth]{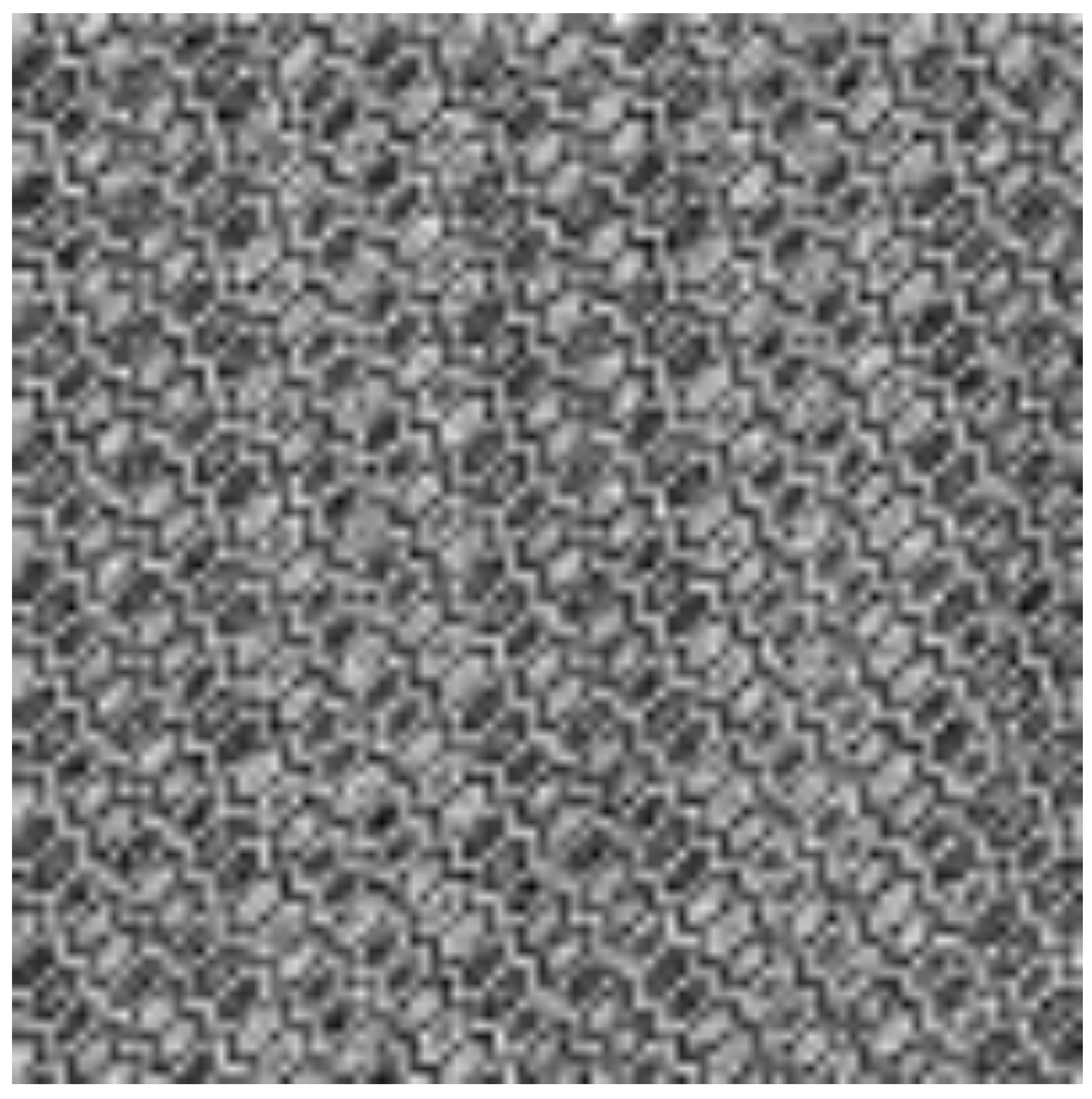}}}\quad\quad
	\subfloat[Original Texture ]{\label{fig:b}{\includegraphics[width=0.22\linewidth]{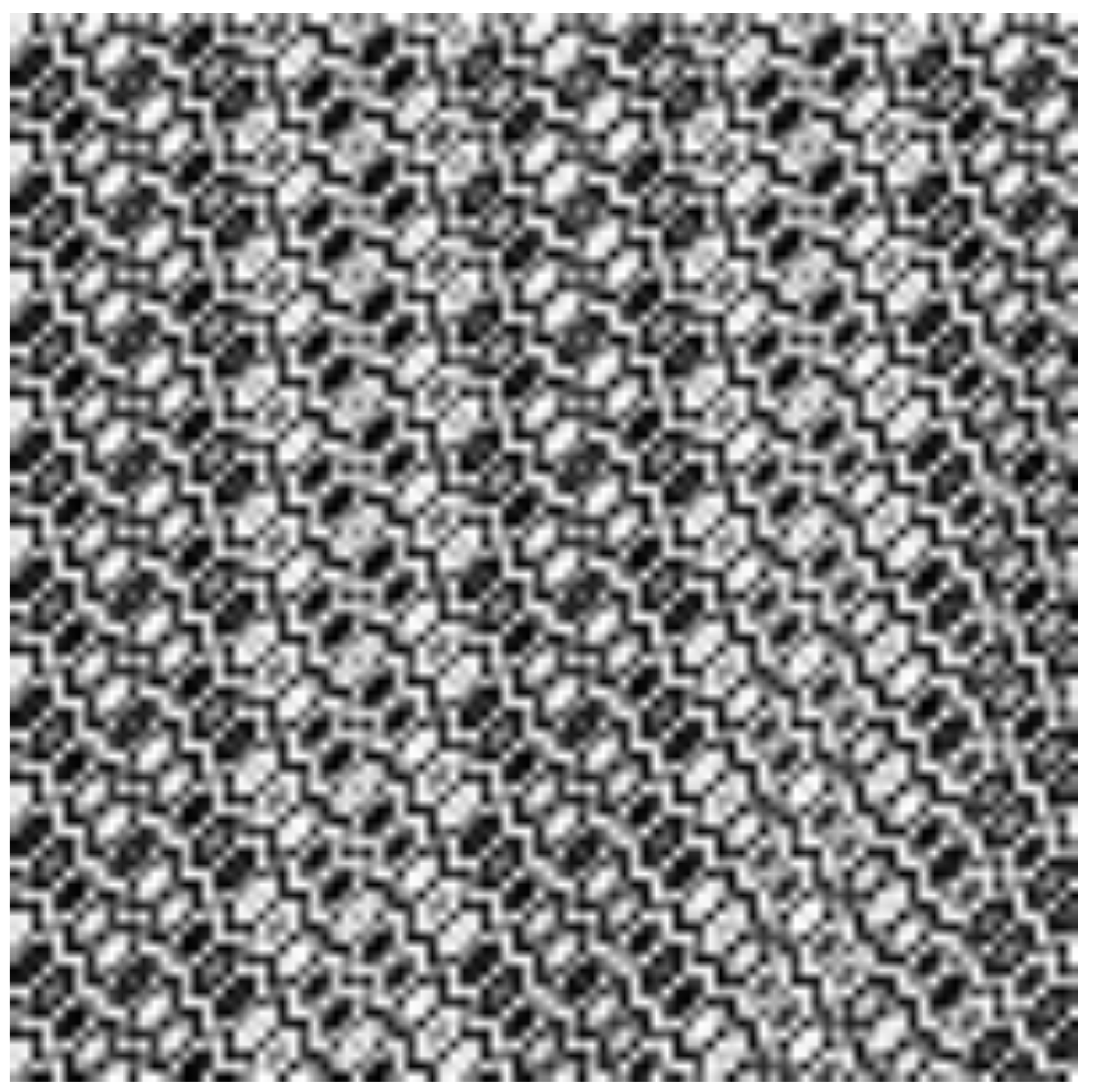}}}\quad\quad 
	\subfloat[Texture estimated \\\text{}\hspace{15pt}using LSEs]{\label{fig:c}{\includegraphics[width=0.22\linewidth]{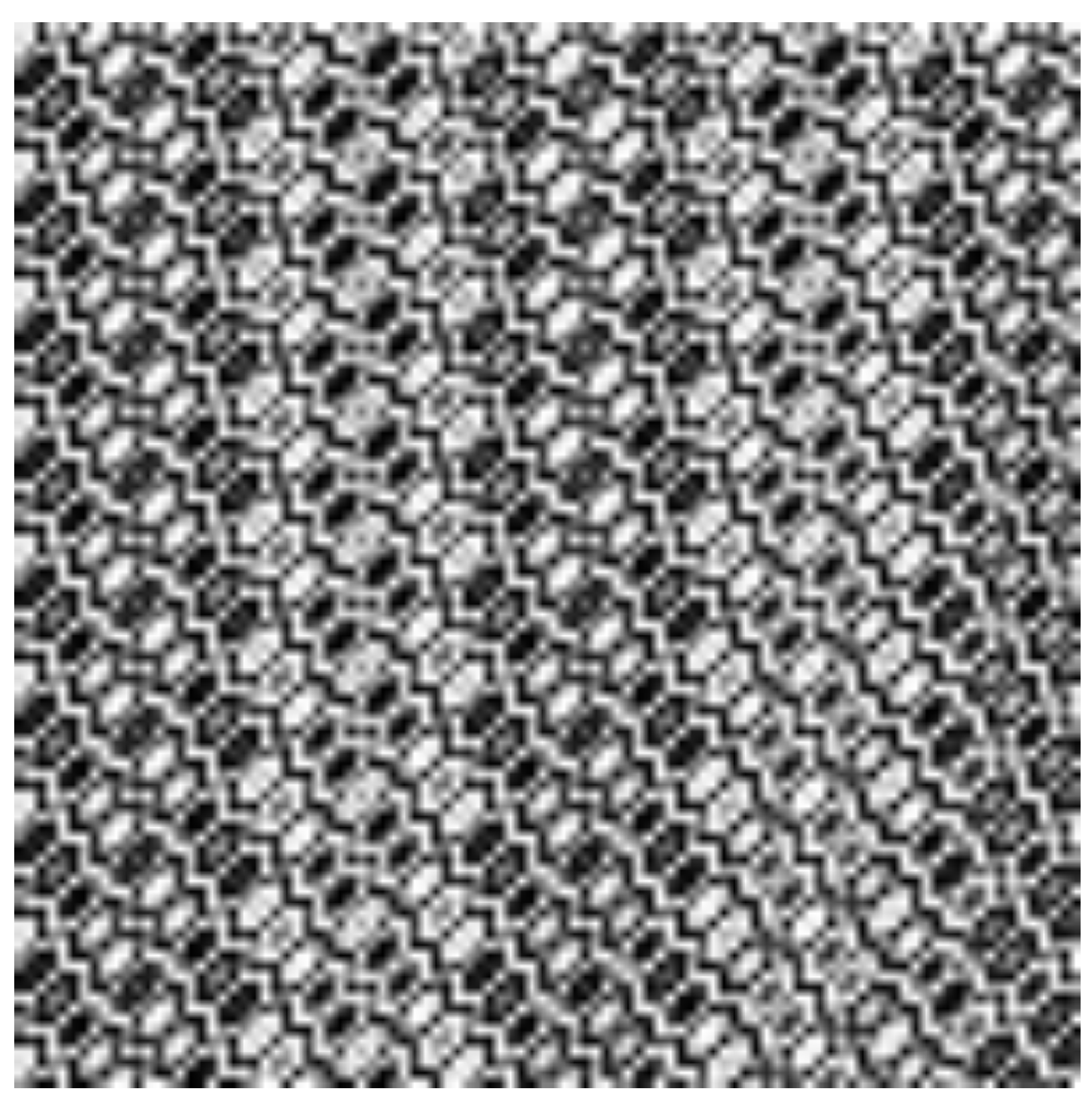}}}\\
		\subfloat[Texture estimated \\\text{}\hspace{15pt}using ALSEs]{\label{fig:d}{\includegraphics[width=0.22\linewidth]{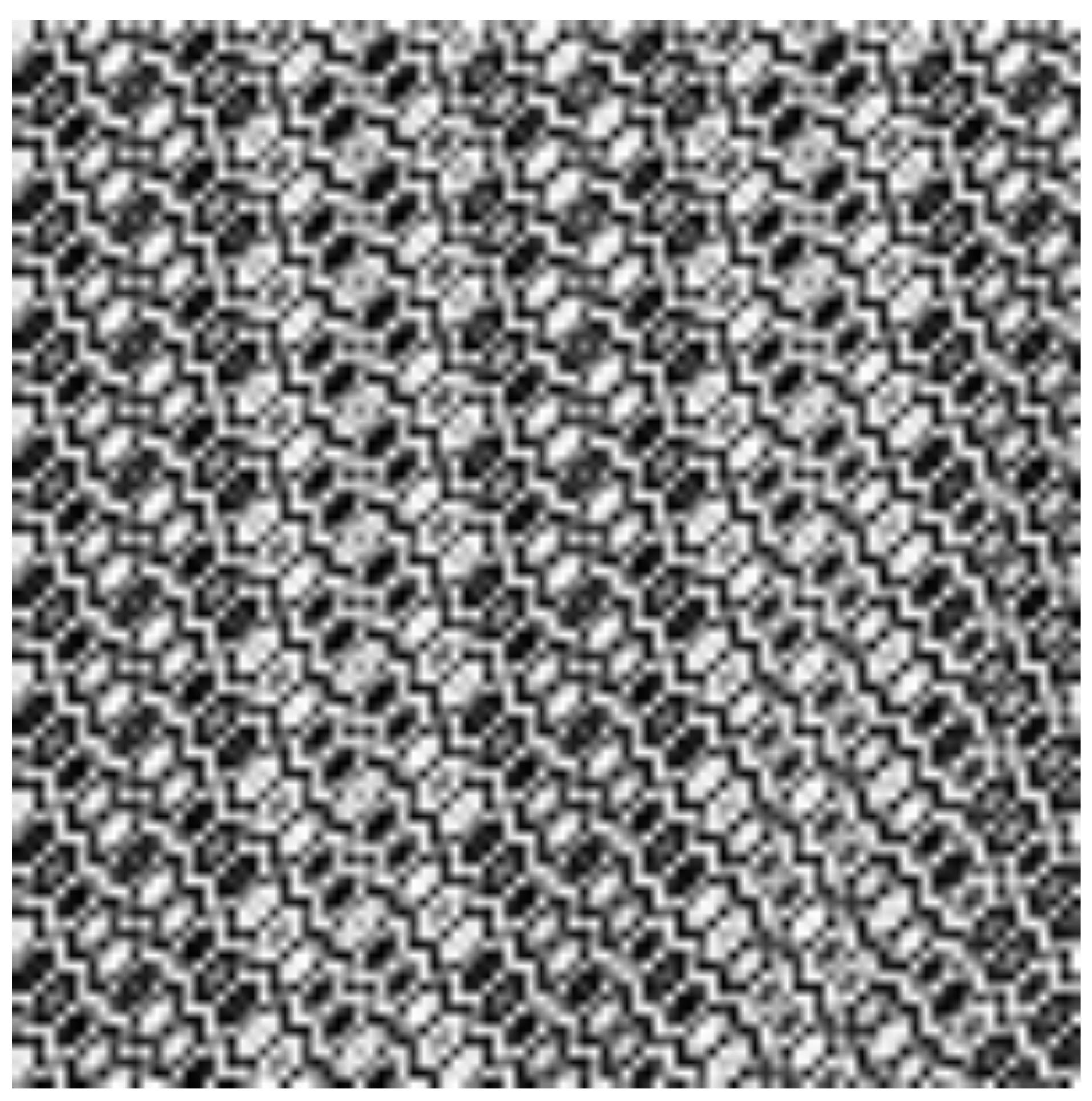}}}\quad\quad
			\subfloat[Texture estimated \\\text{}\hspace{10pt} using  2D-multilag-HAF\\\text{}\hspace{10pt} method]{\label{fig:e}{\includegraphics[width=0.22\linewidth]{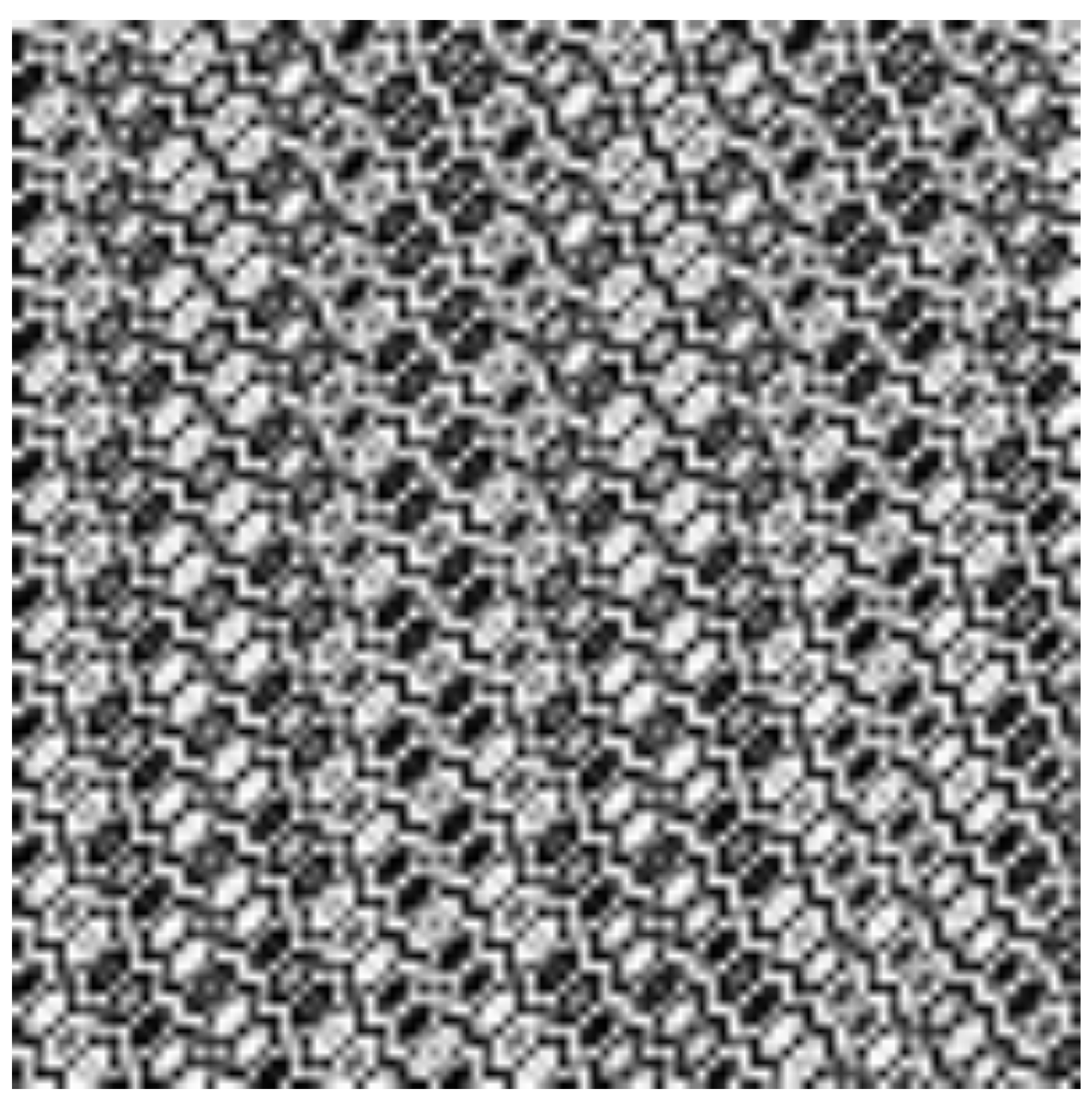}}}\quad\quad
	\subfloat[Texture estimated\\\text{}using proposed method]{\label{fig:f}{\includegraphics[width=0.22\linewidth]{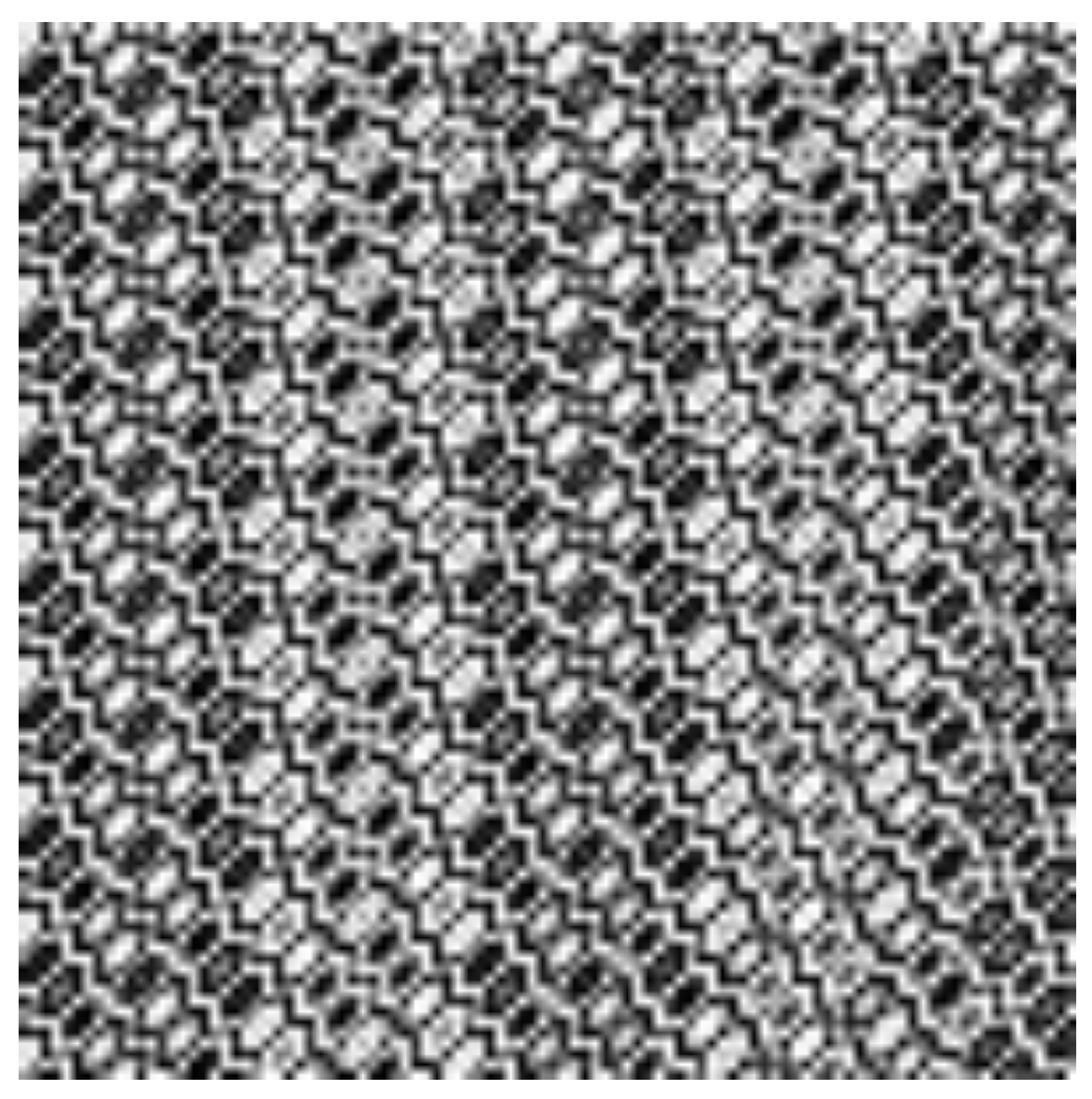}}}%
	\caption{A comparison of estimated textures using LSEs, ALSEs, 2D-multilag-HAF and the proposed estimators.}
	\label{fig:myfig}
\end{figure}

\section{Conclusion}\label{conclu}
The paper proposes a computationally efficient method that produces estimators having the same rate of convergence as the LSEs or ALSEs. The key idea is to develop a strategy by disintegrating the 2D model into several 1D chirp models and then design an optimal estimation method to obtain estimates of the 2D model. We believe that this idea can lead to several other computationally efficient algorithms which can be used to estimate higher order polynomial phase signals’ parameters also. The proposed estimators are not only asymptotically unbiased but also have an asymptotic normal distribution and same rate of convergence as that of the LSEs.  Furthermore, the estimators converge strongly to the true value of the parameters. Extensive numerical simulations firmly support the theoretical results and also unveil the gigantic gap between time required for obtaining proposed estimates and the LSEs. Synthetic data analysis illustrates effectiveness of the proposed estimators to recover 2D gray-scale textures contaminated with noise.
\section*{Acknowledgments}
Part of the work of the third author has been	supported by a research grant from the Science and Engineering Research Board, Government of India.
\section*{Declarations}

\begin{itemize}
	\item Funding: Part of the work of the third author has been	supported by a research grant from the Science and Engineering Research Board, Government of India.
	\item Conflict of interest/Competing interests (check journal-specific guidelines for which heading to use): Not applicable.
	\item Ethics approval: Not applicable.
	\item Consent to participate: Not applicable.
	\item Consent for publication: Not applicable.
	\item Availability of data and materials: Not applicable.
	\item Code availability: The authors can share the code used for simulations, on individual requests. 
	\item Authors' contributions: All authors contributed to  study the conceptualization, methodology, review and editing. Formal analysis of proofs and simulations analysis were performed by Abhinek Shukla. The first draft of the manuscript was written by Abhinek Shukla and all authors commented on revising the previous versions of the manuscript. All authors read and approved the final manuscript. Amit Mitra and Debasis Kundu had supervised the whole research project.
\end{itemize}
\section*{Appendix A} \label{AppA}
\textbf{\large Proof of Theorem \ref{thm_1}:} Given the data matrix $\bm{Y}$, we compute the LSEs of $\alpha^0+n_0\mu^0$ and $\beta^0$ corresponding to $n_0^{th}$ column vector of $\bm{Y}$. We denote the obtained estimators by 	$\widehat{\alpha}_{n_0}$ and $\widehat{\beta}_{n_0}$ to emphasize that these depend on $n_0$. Similarly, for fixed   $m_0^{th}$ row of $\bm{Y}$, we have 
denoted the LSEs of $\gamma^0+m_0\mu^0$ and $\delta^0$ by $\widehat{\gamma}_{m_0}$ and $ \widehat{\delta}_{m_0}$.
  Under assumptions that $X(m_0,n_0)$ are stationary \eqref{err_1} and \eqref{err_2}, see \cite{Nandi_2004},  we have 
\begin{align}
	\widehat{\beta}_{n_0}&= \beta^0+o\bigg(\cfrac{1}{M^2}\bigg), \hspace{15pt}
	\widehat{\alpha}_{n_0}=\alpha^0+n_0\mu^0+o\bigg(\cfrac{1}{M}\bigg),\label{alp_bet_cons}\\
	 \widehat{\gamma}_{m_0}&= \gamma^0+m_0\mu^0+o\bigg(\cfrac{1}{N}\bigg)\label{gam_delt_cons}, \hspace{15pt}
	 \widehat{\delta}_{m_0}=\delta^0+o\bigg(\cfrac{1}{N^2}\bigg).
\end{align}
The final estimator of $\beta^0$  given by $$\widehat{\beta}= \cfrac{\displaystyle\sum_{n_0=1}^{N}\widehat{\beta}_{n_0}}{N}$$ is strongly consistent estimate of $\beta^0$ which is observed by \eqref{alp_bet_cons} and also that $	\widehat{\beta}_{n_0}$ is strongly consistent for $\beta^0$ as $M\xrightarrow{} \infty$. For proof, one may refer to \cite{Lahiri_2015}.  \\~\\ Similarly $\widehat{\delta}=\displaystyle\cfrac{1}{M}\sum_{m_0=1}^{M}\widehat{\delta}_{m_0}$ is strongly consistent estimator  of $\delta^0$.\\
Denote $\bm{\tau}^\top=\begin{bmatrix}
\undermat{\mbox{ $N$ times }}{	o\bigg(\cfrac{1}{M}\bigg)&	\ldots &	o\bigg(\cfrac{1}{M}\bigg)}& \undermat{\mbox{ $M$ times } }{	o\bigg(\cfrac{1}{N}\bigg)&	\ldots&	o\bigg(\cfrac{1}{N}\bigg)}
\end{bmatrix}_{1\times (M+N)}.$\\~\\~\\
We now prove the consistency of the frequency parameter estimators $	\widehat{\alpha},
\widehat{\gamma}$, and that of $\widehat{\mu}$, estimator of the interaction term parameter. From \eqref{alp_bet_cons} and \eqref{gam_delt_cons},  we have the following:
\begin{equation}
	\begin{bmatrix}
		\widehat{\alpha}\\l
		\widehat{\gamma}\\
		\widehat{\mu}
	\end{bmatrix} 
=\Big(\bm{\Gamma}^\top \bm{\Gamma}\Big)^{-1}\bm{\Gamma}^\top \bm{\Lambda}= (\bm{\Gamma}^\top \bm{\Gamma})^{-1}\bm{\Gamma}^\top \bigg(\bm{\Gamma} \begin{bmatrix}
		\alpha^0\\ \gamma^0\\ \mu^0
	\end{bmatrix}+\bm{\tau}\bigg),\\
\end{equation}
where $\bm{\Gamma}^\top=\begin{bmatrix}
	1&1&\cdots &1&0&0&\cdots&0\\
	0&0&\cdots&0&1&1&\cdots &1\\
	1&2&\cdots&N&1&2&\cdots&M
\end{bmatrix}_{3\times(M+N)}$, and $\bm{\Lambda}^\top =\begin{bmatrix}
\widehat{\alpha}_1&
\widehat{\alpha}_2&
\cdots&
\widehat{\alpha}_N&
\widehat{\gamma}_1&
\widehat{\gamma}_2&
\cdots&
\widehat{\gamma}_M
\end{bmatrix}$.\\~\\
This implies that 
\begin{align}\label{alp_consis}
	&\begin{bmatrix}
		\widehat{\alpha}\\
		\widehat{\gamma}\\
		\widehat{\mu}
	\end{bmatrix}= \begin{bmatrix}
		\alpha^0\\ \gamma^0\\ \mu^0
	\end{bmatrix}+(\bm{\Gamma}^\top \bm{\Gamma})^{-1}\begin{bmatrix}
		N\times o\bigg(\cfrac{1}{M}\bigg)\\
		M\times o\bigg(\cfrac{1}{N}\bigg)\\
		\cfrac{N(N+1)}{2}\times o\bigg(\cfrac{1}{M}\bigg)+\cfrac{M(M+1)}{2}\times o\bigg(\cfrac{1}{N}\bigg)
	\end{bmatrix}.
\end{align}
Now we look at the first element $a_1+a_2-a_3$ of the following matrix 
\[(\bm{\Gamma}^\top \bm{\Gamma})^{-1}\begin{bmatrix}
	N\times o\bigg(\cfrac{1}{M}\bigg)\\
	M\times o\bigg(\cfrac{1}{N}\bigg)\\
	\cfrac{N(N+1)}{2}\times o\bigg(\cfrac{1}{M}\bigg)+\cfrac{M(M+1)}{2}\times o\bigg(\cfrac{1}{N}\bigg)
\end{bmatrix}.\]\\~\\
where, 
\begin{flalign*} a_1&=\bigg(MK-\cfrac{M^2(M+1)^2}{4}\bigg) N\times o\bigg(\cfrac{1}{M}\bigg)/\Xi,&\\
		a_2&= \cfrac{MN(M+1)(N+1)}{4}	M\times o\bigg(\cfrac{1}{N}\bigg)/\Xi,&\\
		a_3&= \cfrac{MN(N+1)}{2}\Bigg(\cfrac{N(N+1)}{2}\times o\bigg(\cfrac{1}{M}\bigg)+\cfrac{M(M+1)}{2}\times o\bigg(\cfrac{1}{N}\bigg)\Bigg)/\Xi,&
\end{flalign*}
and $$ \Xi= \cfrac{MN}{12} \bigg(N(N^2-1)+M(M^2-1)\bigg),K=\cfrac{N(N+1)(2N+1)}{6}+\cfrac{M(M+1)(2M+1)}{6}.$$
Now we look at $a_1, a_2$ and $a_3$ individually and compute their limits.
 \begin{flalign*} 
 	a_1&=\bigg(MK-\cfrac{M^2(M+1)^2}{4}\bigg) N\times o\bigg(\cfrac{1}{M}\bigg)/\Xi&\\~\\
 	&= N\times \cfrac{o(1)}{\Xi}\times\bigg(\cfrac{N(N+1)(2N+1)}{6}+\cfrac{M(M+1)(M-1)}{12}\bigg)&\\~\\&=
	\cfrac{o(1)}{M}\times \cfrac{\bigg(2N(N+1)(2N+1)+M(M^2-1)\bigg)}{\bigg(N(N^2-1)+M(M^2-1)\bigg)}.&
\end{flalign*}
$	\mbox{ This implies that } a_1\xrightarrow{a.s.} 0\mbox{  as } \min\{M,N\}\rightarrow \infty.$ Here,  
\begin{flalign*}
	a_2&= \cfrac{MN(M+1)(N+1)}{4}	M\times o\bigg(\cfrac{1}{N}\bigg)/\Xi= o(1)\times \cfrac{M(M+1)}{\bigg(N(N^2-1)+M(M^2-1)\bigg)}.&	
\end{flalign*}
$\mbox{ This  implies that } a_2 \xrightarrow{a.s.} 0\mbox{  as } \min\{M,N\}\rightarrow \infty.$
\begin{flalign*}
	a_3&= \cfrac{MN(N+1)}{2}\Bigg(\cfrac{N(N+1)}{2}\times o\bigg(\cfrac{1}{M}\bigg)+\cfrac{M(M+1)}{2}\times o\bigg(\cfrac{1}{N}\bigg)\Bigg)/\Xi&\\~\\&= 
	\cfrac{(N+1)}{2\Xi}\Bigg(N^2(N+1)\times o(1)+M^2(M+1)\times o(1)\Bigg).&
\end{flalign*}
$\mbox{ This implies that } 
a_3 \xrightarrow{a.s.} 0\mbox{  as } \min\{M,N\}\rightarrow \infty.$
Hence, $\widehat{\alpha}$ is strongly consistent estimate of $\alpha^0$. Similarly strong consistency of  $\widehat{\gamma}$  for  $\gamma^0$ can be derived.\\~\\
Now, the third element of $(\bm{\Gamma}^\top \bm{\Gamma})^{-1}\begin{bmatrix}
	N\times o\bigg(\cfrac{1}{M}\bigg)\\
	M\times o\bigg(\cfrac{1}{N}\bigg)\\
	\cfrac{N(N+1)}{2}\times o\bigg(\cfrac{1}{M}\bigg)+\cfrac{M(M+1)}{2}\times o\bigg(\cfrac{1}{N}\bigg)
\end{bmatrix}$ can be written as $-b_1-b_2+b_3$, where
\begin{flalign*}
	b_1&=\cfrac{MN(N+1)}{2}\times 	N\times o\bigg(\cfrac{1}{M}\bigg)/\Xi=o(1)\cfrac{N(N+1)}{M\big(N(N^2-1)+M(M^2-1)\big)}.&
\end{flalign*}
$\mbox{ This implies that } b_1\xrightarrow{a.s.} 0\mbox{  as } \min\{M,N\}\rightarrow \infty.$
\begin{flalign*}
	b_2&= \cfrac{NM(M+1)}{2}\times 	M\times o\bigg(\cfrac{1}{N}\bigg)/\Xi= 
	o(1)\times \cfrac{M(M+1)}{N\big(N(N^2-1)+M(M^2-1)\big)}.&
\end{flalign*}
$\mbox{ This implies that } b_2\xrightarrow{a.s.} 0\mbox{  as } \min\{M,N\}\rightarrow \infty.$
\begin{flalign*}
	b_3&=MN\times \bigg(\cfrac{N(N+1)}{2}\times o\bigg(\cfrac{1}{M}\bigg)+\cfrac{M(M+1)}{2}\times o\bigg(\cfrac{1}{N}\bigg)\bigg)/\Xi&\\~\\&= \bigg(\cfrac{N(N+1)}{\big(N(N^2-1)+M(M^2-1)\big)}\times o\bigg(\cfrac{1}{M}\bigg)+\cfrac{M(M+1)}{\big(N(N^2-1)+M(M^2-1)\big)}\times o\bigg(\cfrac{1}{N}\bigg)\bigg).&
\end{flalign*}
Hence $
b_3\xrightarrow{a.s.} 0\mbox{  as } \min\{M,N\}\rightarrow \infty$.
Hence, $\widehat{\mu}$ is a strongly consistent estimate of $\mu^0$. The	
strong consistency of estimators of amplitude parameters follows from the continuity mapping theorem.

	\section*{Appendix B}\label{AppB}
\hspace{28pt}	\textbf{\large Proof of Theorem \ref{thm_2}:} Suppose $\bm{\kappa}^\top =(A,B,\alpha,\beta) $ and\\ $\bm{\kappa}^{0^\top }= (A^0(n_0),B^0(n_0),\alpha^0+n_0\mu^0,\beta^0)$. We define sum of squares as follows: \[Q_{n_0}(\bm{\kappa})= \displaystyle\sum_{m_0=1}^M\bigg(y(m_0,n_0)-A\cos(\alpha m_0+\beta m_0^2)-B\sin(\alpha m_0+\beta m_0^2)\bigg)^2.\]
Let  $\widehat{\bm{\kappa}}$ be the minimizer of $Q_{n_0}(\bm{\kappa})$, then using Taylor Series expansion on the first derivative vector $Q_{n_0}^{\prime}(\bm{\kappa})$ around the point $\bm{\kappa}^{0}$, we get:
\begin{equation}\label{tayl_ser}
	\bm{Q}_{n_0}^{\prime}(\widehat{\bm{\kappa}})-\bm{Q}_{n_0}^{\prime}(\bm{\kappa}^0) = \bm{Q}_{n_0}^{\prime\prime}(\breve{\bm{\kappa}}) (\widehat{\bm{\kappa}}-\bm{\kappa}^0),
\end{equation}  
where $\breve{\bm{\kappa}}$ is a point between $\widehat{\bm{\kappa}}$ and $\bm{\kappa}^0$. Also $\bm{Q}_{n_0}^{\prime}(\widehat{\bm{\kappa}})=0$ as $\widehat{\bm{\kappa}}$ is LSE of $\bm{\kappa}^0$.\\ 
Let us denote  $\bm{D}_1^{-1}=diag(M^{1/2},M^{1/2},M^{3/2},M^{5/2})$. Then on multiplying $\bm{D}_1^{-1}$ both sides of equation \eqref{tayl_ser}, it gives: 				
\begin{align}\label{final_asmp}
	-\big[\bm{D_1}\bm{Q}_{n_0}^{\prime\prime}(\breve{\bm{\kappa}})\bm{D_1}\big]^{-1}\bm{\Sigma}_{n_0} \displaystyle \bm{\Sigma}_{n_0}^{-1}\bm{D_1}\bm{Q}_{n_0}^{\prime}(\bm{\kappa}^0),& = \bm{D_1}^{-1}(\widehat{\bm{\kappa}}-\bm{\kappa}^0)\\\nonumber~\\
	\mbox{where } \lim \limits_{M\rightarrow \infty}\big[\bm{D_1}\bm{Q}_{n_0}^{\prime\prime}(\breve{\bm{\kappa}})\bm{D_1}\big]^{-1}\bm{\Sigma}_{n_0} &= I_{4\times 4}\nonumber,
\end{align}
and 
\begin{equation*}
	\displaystyle\displaystyle \bm{\Sigma}_{n_0}^{-1}= \cfrac{2}{A^{0^2}+B^{0^2}}	
	\begin{bmatrix}
		\cfrac{A^{0^2}(n_0)+9B^{0^{2}}(n_0)}{2}&-4A^0(n_0)B^0(n_0)&-18B^0(n_0)&15B^0(n_0)\\
		-4A^0(n_0)B^0(n_0)&\cfrac{9A^{0^2}(n_0)+B^{0^{2}}(n_0)}{2}&18A^0(n_0)&-15A^0(n_0)\\
		-18B^0(n_0)&18A^0(n_0)&96&-90\\
		15B^0(n_0)&-15A^0(n_0)&-90&90
	\end{bmatrix}.
\end{equation*}
The expression of $\bm{\Sigma}_{n_0}^{-1}$ can be obtained by proof of Theorem 2 shown in \cite{Lahiri_2015}.
By using equation \eqref{final_asmp} and above expression of $\bm{\Sigma}_{n_0}^{-1}$, we get following asymptotically equivalent (a.e.) expression of the  third element of  vector $\bm{D_1}^{-1}(\widehat{\bm{\kappa}}-\bm{\kappa}^0)$ as:
\begin{align}\label{alp_asy}
	\begin{bmatrix}
		M^{3/2} \big(\widehat{\alpha}_{n_0}-(\alpha^0+n_0\mu^0)\big) 
	\end{bmatrix}&\myeq  \cfrac{4}{A^{0^2}+B^{0^2}}\Bigg[18\displaystyle\sum_{m} \cfrac{\eta(m_0,n_0)}{\sqrt{M}}-96 \displaystyle\sum_{m_0} \cfrac{m_0\eta(m_0,n_0)}{M\sqrt{M}}+90\displaystyle\sum_{m_0} \cfrac{m_0^2\eta(m_0,n_0)}{M^2\sqrt{M}}\Bigg], 
\end{align}		
where, $\eta(m_0,n_0)=X(m_0,n_0)\big(A^0\sin\phi(m_0,n_0,\bm{\xi}^0)-B^0\cos\phi(m_0,n_0,\bm{\xi}^0)\big)$, \\$\phi(m_0,n_0,\bm{\xi}^0)= \alpha^0m_0+\beta^0m_0^2+\gamma^0n_0+\delta^0n_0^2+\mu^0 m_0n_0$. We have used these notations for brevity. Similarly expressions of fourth, fifth and sixth element of $\bm{D_1}^{-1}(\widehat{\bm{\kappa}}-\bm{\kappa}^0)$ can be written as in equations \eqref{bet_asy}, \eqref{gam_asy} and \eqref{delt_asy} respectively:
\begin{align}\label{bet_asy} 
	\begin{bmatrix}
		M^{5/2}\big(\widehat{\beta}_{n_0}-\beta^0\big) 
	\end{bmatrix}&\myeq  \cfrac{4}{A^{0^2}+B^{0^2}}\Bigg[-15\displaystyle\sum_{m_0} \cfrac{\eta(m_0,n_0)}{\sqrt{M}}+90 \displaystyle\sum_{m_0} \cfrac{m_0\eta(m_0,n_0)}{M\sqrt{M}}-90\displaystyle\sum_{m_0} \cfrac{m_0^2\eta(m_0,n_0)}{M^2\sqrt{M}}\Bigg] ,
\end{align}
\begin{align}\label{gam_asy}
	\begin{bmatrix}
		N^{3/2} \big(\widehat{\gamma}_{m_0}-(\gamma^0+m_0\mu^0)\big) 
	\end{bmatrix}&\myeq  \cfrac{4}{A^{0^2}+B^{0^2}}\Bigg[18\displaystyle\sum_{n_0} \cfrac{\eta(m_0,n_0)}{\sqrt{N}}-96 \displaystyle\sum_{n_0} \cfrac{n_0\eta(m_0,n_0)}{N\sqrt{N}}+90\displaystyle\sum_{n_0} \cfrac{n_0^2\eta(m_0,n_0)}{N^2\sqrt{N}}\Bigg], 
\end{align}

\begin{align}\label{delt_asy}
	\begin{bmatrix}
		N^{5/2}\big(\widehat{\delta}_{m_0}-\delta^0\big) 
	\end{bmatrix}&\myeq  \cfrac{4}{A^{0^2}+B^{0^2}}\Bigg[-15\displaystyle\sum_{n_0} \cfrac{\eta(m_0,n_0)}{\sqrt{N}}+90 \displaystyle\sum_{n_0} \cfrac{n_0\eta(m_0,n_0)}{N\sqrt{N}}-90\displaystyle\sum_{n_0} \cfrac{n_0^2\eta(m_0,n_0)}{N^2\sqrt{N}}\Bigg] .
\end{align}
From equations   \eqref{bet_asy} and \eqref{delt_asy} above, we can see that estimators $\widehat{\beta}$ and $\widehat{\delta}$  of  $\beta^0$ and $\delta^0$  are  asymptotically equivalent to the LSEs as they have same asymptotic variances by applying central limit theorem for stationary linear processes, see \cite{Fuller_1976}. So now, remaining is to show asymptotic properties of estimators $  (	\widehat{\alpha},
\widehat{\gamma},
\widehat{\mu})^\top$.  \\~\\
The expression of proposed estimators of $(\alpha^0,\gamma^0, \mu^0)^\top$  obtained is:
\begin{equation} \label{matrix_final}
	\begin{bmatrix}
		\widehat{\alpha}\\
		\widehat{\gamma}\\
		\widehat{\mu}
	\end{bmatrix} = \begin{bmatrix}
		\bigg(MK-\cfrac{M^2\big(M+1\big)^2}{4}\bigg)c_{\alpha}+ \cfrac{MN(M+1)(N+1)}{4}c_{\gamma}-\cfrac{MN(N+1)}{2}c_{\alpha\gamma}\\~\\
		\cfrac{MN(M+1)(N+1)}{4}c_{\alpha}+\bigg(KN-\cfrac{N^2(N+1)^2}{4}\bigg)c_{\gamma}-\cfrac{NM(M+1)}{2}c_{\alpha\gamma}\\~\\
		-\cfrac{MN(N+1)}{2}c_{\alpha}-\cfrac{NM(M+1)}{2}c_{\gamma}+MNc_{\alpha\gamma}
	\end{bmatrix}/| \bm{\Gamma}^\top \bm{\Gamma}|,
\end{equation}
\[c_{\alpha}=\displaystyle\sum_{n_0=1}^{N}\widehat{\alpha}_{n_0}, c_{\gamma}=\displaystyle\sum_{m_0=1}^M\widehat{\gamma}_{m_0}, c_{\alpha\gamma} = \displaystyle\sum_{n_0=1}^{N}n_0\widehat{\alpha}_{n_0}+\displaystyle\sum_{m_0=1}^Mm_0\widehat{\gamma}_{m_0},\]
\[ K=\cfrac{N(N+1)(2N+1)}{6}+\cfrac{M(M+1)(2M+1)}{6},  \big|\bm{\Gamma}^\top \bm{\Gamma}\big|= \cfrac{MN}{12} \bigg(N(N^2-1)+M(M^2-1)\bigg).\] 
From \eqref{matrix_final}, we get:
\begin{align*}
M^{3/2}N^{1/2}\big(\widehat{\alpha}-\alpha^0\big)&= \bigg(\cfrac{2N(N+1)(2N+1)+M(M^2-1)}{N(N^2-1)+M(M^2-1)}\bigg)\cfrac{1}{\sqrt{N}}\displaystyle\sum_{n_0=1}^{N}M^{3/2}\big(\widehat{\alpha}_{n_0}-(\alpha^0+n_0\mu^0)\big)\\&+
\bigg(\cfrac{3M^2(M+1)}{N(N^2-1)+M(M^2-1)}\bigg)\cfrac{1}{\sqrt{M}}\displaystyle\sum_{m_0=1}^MN^{3/2}\big(\widehat{\gamma}_{m_0}-(\gamma^0+m_0\mu^0)\big)\\&-
\bigg(\cfrac{6M^{3/2}N^{1/2}(N+1)}{N(N^2-1)+M(M^2-1)}\bigg)\bigg[\displaystyle\sum_{n_0=1}^{N}n_0\big(\widehat{\alpha}_{n_0}-(\alpha^0+n_0\mu^0)\big)\\&+\displaystyle\sum_{m_0=1}^Mm_0\big(\widehat{\gamma}_{m_0}-(\gamma^0+m_0\mu^0)\big)\bigg].
\end{align*}	
For sufficiently large $M$ and $N$, we have:		
\begin{align}\label{expr_alp}
M^{3/2}N^{1/2}\big(\widehat{\alpha}-\alpha^0\big)&=\bigg(\cfrac{4N^3+M^3}{N^3+M^3}\bigg) \cfrac{1}{\sqrt{N}}\displaystyle\sum_{n_0=1}^{N}M^{3/2}\big(\widehat{\alpha}_{n_0}-(\alpha^0+n_0\mu^0)\big)\nonumber\\&+\bigg(\cfrac{3M^3}{N^3+M^3}\bigg)\cfrac{1}{\sqrt{M}}\displaystyle\sum_{m_0=1}^MN^{3/2}\big(\widehat{\gamma}_{m_0}-(\gamma^0+m_0\mu^0)\big)\nonumber\\&-\bigg(\cfrac{6N^3}{N^3+M^3}\bigg)\cfrac{1}{N^{3/2}}\displaystyle\sum_{n_0=1}^{N}n_0M^{3/2}\big(\widehat{\alpha}_{n_0}-(\alpha^0+n_0\mu^0)\big)\nonumber\\&-
\bigg(\cfrac{6M^3}{N^3+M^3}\bigg)\cfrac{1}{M^{3/2}}\displaystyle\sum_{m_0=1}^Mm_0N^{3/2}\big(\widehat{\gamma}_{m_0}-(\gamma^0+m_0\mu^0)\big).
\end{align}
Asymptotic normality of the estimators for $M=N\rightarrow\infty$ with given rates of convergence follows by applying central limit theorem for stationary processes \cite{Fuller_1976}.\\  

We now present some important results which will be used to find the asymptotic variance-covariance matrix of the proposed estimators of non-linear parameters. Using equations  \eqref{alp_asy} and  \eqref{gam_asy} in the paper, we have the following observations:
\begin{enumerate}
	\item $Asy Var \bigg(\cfrac{1}{2\sqrt{N}}\displaystyle\sum_{n_0=1}^{N}M^{3/2}\big(\widehat{\alpha}_{n_0}-(\alpha^0+n_0\mu^0)\big)\bigg)=\cfrac{c\sigma^2 96}{A^{0^2}+B^{0^2}}$,
	\item $Asy Var \bigg(\cfrac{1}{2\sqrt{M}}\displaystyle\sum_{m_0=1}^MN^{3/2}\big(\widehat{\gamma}_{m_0}-(\gamma^0+m_0\mu^0)\big)\bigg)=\cfrac{c\sigma^296}{A^{0^2}+B^{0^2}}$,
	\item $Asy Var \bigg(\cfrac{1}{2N^{3/2}}\displaystyle\sum_{n_0=1}^{N}n_0M^{3/2}\big(\widehat{\alpha}_{n_0}-(\alpha^0+n_0\mu^0)\big)\bigg)=\cfrac{c\sigma^232}{A^{0^2}+B^{0^2}}$,\label{doubt_1}
	\item  $Asy Var \bigg(\cfrac{1}{2M^{3/2}}\displaystyle\sum_{m_0=1}^Mm_0N^{3/2}\big(\widehat{\gamma}_{m_0}-(\gamma^0+m_0\mu^0)\big)\bigg)=\cfrac{c\sigma^232}{A^{0^2}+B^{0^2}}$,\label{doubt_2}
	\item  $Asy Covar \bigg(\cfrac{1}{2\sqrt{N}}\displaystyle\sum_{n_0=1}^{N}M^{3/2}\big(\widehat{\alpha}_{n_0}-(\alpha^0+n_0\mu^0)\big),\cfrac{1}{2\sqrt{M}}\displaystyle\sum_{m_0=1}^MN^{3/2}\big(\widehat{\gamma}_{m_0}-(\gamma^0+m_0\mu^0)\big)\bigg)=0$,
	\item $Asy Covar \bigg(\cfrac{1}{2\sqrt{N}}\displaystyle\sum_{n_0=1}^{N}M^{3/2}\big(\widehat{\alpha}_{n_0}-(\alpha^0+n_0\mu^0)\big),\cfrac{1}{2N^{3/2}}\displaystyle\sum_{n_0=1}^{N}n_0M^{3/2}\big(\widehat{\alpha}_{n_0}-(\alpha^0+n_0\mu^0)\big)\bigg)=\cfrac{c\sigma^248}{A^{0^2}+B^{0^2}}$,
	\item $Asy Covar \bigg(\cfrac{1}{2\sqrt{N}}\displaystyle\sum_{n_0=1}^{N}M^{3/2}\big(\widehat{\alpha}_{n_0}-(\alpha^0+n_0\mu^0)\big),\cfrac{1}{2M^{3/2}}\displaystyle\sum_{m_0=1}^Mm_0N^{3/2}\big(\widehat{\gamma}_{m_0}-(\gamma^0+m_0\mu^0)\big)\bigg)=0$,
	\item $Asy Covar \bigg(\cfrac{1}{2\sqrt{M}}\displaystyle\sum_{m_0=1}^MN^{3/2}\big(\widehat{\gamma}_{m_0}-(\gamma^0+m_0\mu^0)\big),\cfrac{1}{2N^{3/2}}\displaystyle\sum_{n_0=1}^{N}n_0M^{3/2}\big(\widehat{\alpha}_{n_0}-(\alpha^0+n_0\mu^0)\big)\bigg)=0$,
	\item $Asy Covar \bigg(\cfrac{1}{2\sqrt{M}}\displaystyle\sum_{m_0=1}^MN^{3/2}\big(\widehat{\gamma}_{m_0}-(\gamma^0+m_0\mu^0)\big),\cfrac{1}{2M^{3/2}}\displaystyle\sum_{m_0=1}^Mm_0N^{3/2}\big(\widehat{\gamma}_{m_0}-(\gamma^0+m_0\mu^0)\big)\bigg)=\cfrac{c\sigma^248}{A^{0^2}+B^{0^2}}$,
	\item  $Asy Covar \bigg(\cfrac{1}{2N^{3/2}}\displaystyle\sum_{n_0=1}^{N}n_0M^{3/2}\big(\widehat{\alpha}_{n_0}-(\alpha^0+n_0\mu^0)\big),\cfrac{1}{2M^{3/2}}\displaystyle\sum_{m_0=1}^Mm_0N^{3/2}\big(\widehat{\gamma}_{m_0}-(\gamma^0+m_0\mu^0)\big)\bigg)=\cfrac{c\sigma^2}{2(A^{0^2}+B^{0^2})}$.
\end{enumerate}\text{}\\~\\\text{}
where $c=\displaystyle{\sum_{i=-\infty}^{\infty}\sum_{j=-\infty}^{\infty}}a^2(i,j)$.\\~\\
Using the above results in  equation \eqref{expr_alp} from the paper, we get asymptotic variance-covariance matrix of\\~\\ $\begin{bmatrix}
	M^{3/2}N^{1/2}(\widehat{\alpha}-\alpha^0)\\
	N^{3/2}M^{1/2}(\widehat{\gamma}-\gamma^0)\\
	M^{3/2}N^{3/2}(\widehat{\mu}-\mu^0)
\end{bmatrix}$ as:
\(
\cfrac{c\sigma^2}{(A^{0^2}+B^{0^2})}\begin{bmatrix}
	996&612&-1224\\
	612	&996&-1224\\
	-1224	&-1224&2448
\end{bmatrix}	.
\)\\~\\
From  \eqref{alp_asy},  \eqref{bet_asy} and \eqref{gam_asy} equations of the paper, it is further observed that:
\begin{enumerate}
	\item 	$AsyCovar\bigg(	M^{5/2}N^{1/2}\big(\widehat{\beta}-\beta^0\big), M^{1/2}N^{5/2}\big(\widehat{\delta}-\delta^0\big)\bigg)=0$,
	\item 	$AsyCovar\bigg(	M^{5/2}N^{1/2}\big(\widehat{\beta}-\beta^0\big),\cfrac{1}{\sqrt{N}}\displaystyle\sum_{n_0=1}^{N}M^{3/2}\big(\widehat{\alpha}_{n_0}-(\alpha^0+n_0\mu^0)\big)\bigg)=\cfrac{-360c\sigma^2}{A^{0^2}+B^{0^2}}$,
	\item 	$AsyCovar\bigg(	M^{5/2}N^{1/2}\big(\widehat{\beta}-\beta^0\big),\cfrac{1}{\sqrt{M}}\displaystyle\sum_{m_0=1}^MN^{3/2}\big(\widehat{\gamma}_{m_0}-(\gamma^0+m_0\mu^0)\big)\bigg)=0$,
	\item 	$AsyCovar\bigg(	M^{5/2}N^{1/2}\big(\widehat{\beta}-\beta^0\big),\cfrac{1}{N\sqrt{N}}\displaystyle\sum_{n_0=1}^{N}M^{3/2}n_0\big(\widehat{\alpha}_{n_0}-(\alpha^0+n_0\mu^0)\big)\bigg)=\cfrac{-180c\sigma^2}{A^{0^2}+B^{0^2}}$,
	\item 	$AsyCovar\bigg(	M^{5/2}N^{1/2}\big(\widehat{\beta}-\beta^0\big),\cfrac{1}{M\sqrt{M}}\displaystyle\sum_{m_0=1}^Mm_0N^{3/2}\big(\widehat{\gamma}_{m_0}-(\gamma^0+m_0\mu^0)\big)\bigg)=0.$
\end{enumerate}
Similar results can be derived for $M^{1/2}N^{5/2}\big(\widehat{\delta}-\delta^0\big)$.\\~\\
Asymptotic variance-covariance matrix of the proposed estimators of non-linear parameters is given by:
\begin{equation}
\cfrac{c\sigma^2}{(A^{0^2}+B^{0^2})}\begin{bmatrix}
	996&-360&612&0&-1224\\
	-360&360&0&0&0\\
	612&0&996&-360&-1224\\
	0&0&-360&360&0\\
	-1224&0&-1224&0&2448
\end{bmatrix}.	
\end{equation}
Next, we derive the asymptotics of amplitude estimators, please recall that by using Taylor series expansion of $\cos\phi(m_0,n_0,\widehat{\bm{\xi}})$ around the point $\bm{\xi}^0$, we can write: \[\cos\widehat{\phi}-\cos\phi^0=-\sin\breve{\phi}\big(\widehat{\bm{\xi}}-\bm{\xi}^0\big)^\top \begin{bmatrix}
	m_0\\m_0^2\\n_0\\n_0^2\\m_0n_0
\end{bmatrix}.\]
For brevity, we have denoted $\cos\phi(m_0,n_0,\widehat{\bm{\xi}})$ by $\cos\widehat{\phi}$, $\cos\phi(m_0,n_0,\widehat{\bm{\xi}})$ by $\cos\phi^0$, and $\sin\phi(m_0,n_0,\breve{\bm{\xi}})$  by  $\sin\breve{\phi}$ 
, where $\breve{\bm{\xi}}$ is a point lying between $\widehat{\bm{\xi}}$ and $\bm{\xi}^0$.\\~\\
Now consider first element of  the following vector,  
$$\sqrt{MN}\Bigg(\begin{bmatrix}
	\cfrac{2}{{MN}}\displaystyle \sum_{m_0=1}^M\sum_{n_0=1}^{N}y(m_0,n_0)\cos\widehat{\phi} \\
	\cfrac{2}{{MN}}\displaystyle \sum_{m_0=1}^M\sum_{n_0=1}^{N}y(m_0,n_0)\sin\widehat{\phi} 
\end{bmatrix}-\begin{bmatrix}
	A^0\\B^0
\end{bmatrix}\Bigg) ,$$ we get:
\begin{align}\label{asymp_proof_step_1}
	\sqrt{MN}\bigg(	\cfrac{2}{{MN}}\displaystyle \sum_{m_0=1}^M\sum_{n_0=1}^{N}y(m_0,n_0)\cos{\phi^0}-\cfrac{2}{{MN}}\displaystyle \sum_{m_0=1}^M\sum_{n_0=1}^{N}y(m_0,n_0)\sin{\breve{\phi}}R(m_0,n_0)-A^0\bigg),	
\end{align}
where $R(m_0,n_0)=\big(\widehat{\bm{\xi}}-\bm{\xi}^0\big)^\top \begin{bmatrix}
	m_0\\m_0^2\\n_0\\n_0^2\\m_0n_0
\end{bmatrix}$, and second element of the above amplitude vector  can be written as:
\begin{align}\label{asymp_proof_step_2_B}
	\sqrt{MN}\bigg(	\cfrac{2}{{MN}}\displaystyle \sum_{m_0=1}^M\sum_{n_0=1}^{N}y(m_0,n_0)\sin{\phi^0}+\cfrac{2}{{MN}}\displaystyle \sum_{m_0=1}^M\sum_{n_0=1}^{N}y(m_0,n_0)\sin{\breve{\phi}}R(m_0,n_0)-B^0\bigg).	
\end{align}
Now let us look at the first and the last term of equation \eqref{asymp_proof_step_1} and putting value of $y(m_0,n_0)$ from the model \eqref{model_2D},
\begin{align}
	&	\sqrt{MN}\bigg(	\cfrac{2}{{MN}}\displaystyle \sum_{m_0=1}^M\sum_{n_0=1}^{N}y(m_0,n_0)\cos{\phi^0}-A^0\bigg)\nonumber\\&=
	\sqrt{MN}\bigg(	\cfrac{2}{{MN}}\displaystyle \sum_{m_0=1}^M\sum_{n_0=1}^{N}A^0\cos^2\phi^0-A^0+ \cfrac{2}{{MN}}\displaystyle \sum_{m_0=1}^M\sum_{n_0=1}^{N}B^0\sin\phi^0\cos\phi^0+\cfrac{2}{{MN}}\displaystyle \sum_{m_0=1}^M\sum_{n_0=1}^{N}X(m_0,n_0)\cos\phi^0\bigg)\nonumber\\&\myeq
	\cfrac{2}{\sqrt{MN}}\displaystyle \sum_{m_0=1}^M\sum_{n_0=1}^{N}X(m_0,n_0)\cos\phi^0.
\end{align}\\
The above result has been obtained from a famous number theory conjecture by \cite{Montgomery}.\\
 In second term of \eqref{asymp_proof_step_1}, $R(m_0,n_0)$ is a sum of five terms, so now consider first term of \\ $\cfrac{2}{{MN}}\displaystyle \sum_{m_0=1}^M\sum_{n_0=1}^{N}y(m_0,n_0)\sin{\breve{\phi}}R(m_0,n_0)$,
\begin{align*}
	&	\cfrac{2}{{\sqrt{MN}}}\displaystyle \sum_{m_0=1}^M\sum_{n_0=1}^{N}y(m_0,n_0)\sin{\breve{\phi}}m_0(\widehat{\alpha}-\alpha^0)\nonumber\\&=
	\cfrac{2}{{\sqrt{MN}}}\displaystyle \sum_{m_0=1}^M\sum_{n_0=1}^{N}A^0\cos\phi^0\sin{\breve{\phi}}m_0(\widehat{\alpha}-\alpha^0)+
	\cfrac{2}{{\sqrt{MN}}}\displaystyle \sum_{m_0=1}^M\sum_{n_0=1}^{N}B^0\sin\phi^0\sin{\breve{\phi}}m_0(\widehat{\alpha}-\alpha^0)\nonumber\\&+
	\cfrac{2}{{\sqrt{MN}}}\displaystyle \sum_{m_0=1}^M\sum_{n_0=1}^{N}X(m_0,n_0)\sin{\breve{\phi}}m_0(\widehat{\alpha}-\alpha^0) \end{align*} 
\begin{align*} 
&=	\bigg(\cfrac{2}{{{M^2N}}}\displaystyle \sum_{m_0=1}^M\sum_{n_0=1}^{N}A^0\cos\phi^0\sin{\breve{\phi}}m_0\bigg)M\sqrt{MN}(\widehat{\alpha}-\alpha^0)+\\&\hspace{20pt}
	\bigg(\cfrac{2}{{{M^2N}}}\displaystyle \sum_{m_0=1}^M\sum_{n_0=1}^{N}B^0\sin\phi^0\sin{\breve{\phi}}m_0\bigg)M\sqrt{MN}(\widehat{\alpha}-\alpha^0)+\nonumber\\&\hspace{20pt}
	\bigg(\cfrac{2}{{{M^2N}}}\displaystyle \sum_{m_0=1}^M\sum_{n_0=1}^{N}X(m_0,n_0)\sin{\breve{\phi}}m_0\bigg)M\sqrt{MN}(\widehat{\alpha}-\alpha^0)\\&\myeq 	\bigg(\cfrac{2}{{{M^2N}}}\displaystyle \sum_{m_0=1}^M\sum_{n_0=1}^{N}B^0\sin\phi^0\sin{\breve{\phi}}m_0\bigg)M\sqrt{MN}(\widehat{\alpha}-\alpha^0), \mbox{\hspace{20pt}by using Proposition 1 of   \cite{Lahiri_2017}.}	
\end{align*}

So, finding the asymptotic distribution of \eqref{asymp_proof_step_1} boils down to finding asymptotic distribution of $$	\cfrac{2}{\sqrt{MN}}\displaystyle \sum_{m_0=1}^M\sum_{n_0=1}^{N}X(m_0,n_0)\cos\phi^0-\bigg(\cfrac{2}{{\sqrt{MN}}}\displaystyle \sum_{m_0=1}^M\sum_{n_0=1}^{N}B^0\sin\phi^0\sin{\breve{\phi}}\bigg)R(m_0,n_0)	,$$ which is further asymptotically equivalent to:
\begin{align}\label{final_expression_amplitude}
	&\cfrac{2}{\sqrt{MN}}\displaystyle \sum_{m_0=1}^M\sum_{n_0=1}^{N}X(m_0,n_0)\cos\phi^0-B^0\bigg(\cfrac{1}{2}M^{3/2}N^{1/2}(\widehat{\alpha}-\alpha^0)+\cfrac{1}{3}M^{5/2}N^{1/2}(\widehat{\beta}-\beta^0)\nonumber\\&+\cfrac{1}{2}N^{3/2}M^{1/2}(\widehat{\gamma}-\gamma^0)+\cfrac{1}{3}N^{5/2}M^{1/2}(\widehat{\delta}-\delta^0)+\cfrac{1}{4}M^{3/2}N^{3/2}(\widehat{\mu}-\mu^0)\bigg).
\end{align}
Asymptotic normality of amplitude estimators is thus proved by equation \eqref{final_expression_amplitude}. Now we need to derive the expression of their asymptotic variances.\\~\\
After lengthy calculations, we get the asymptotic variance of $\widehat{A}$ as follows: $$\cfrac{c\sigma^2}{(A^{0^2}+B^{0^2})}(2A^{0^2}+187B^{0^2}).$$ Similarly by calculating other terms too, we get complete variance co-variance matrix $\bm{\Sigma}$ as mentioned in Theorem \ref{thm_2}.\\~\\
Hence the result.

\end{document}